\documentclass[sn-mathphys]{sn-jnl}

\jyear{2021}%

\theoremstyle{thmstyleone}%
\newtheorem{theorem}{Theorem}
\newtheorem{proposition}[theorem]{Proposition}%
\newtheorem{lemma}[theorem]{Lemma}

\theoremstyle{thmstyletwo}%
\newtheorem{remark}{Remark}%

\theoremstyle{thmstylethree}%
\newtheorem{definition}{Definition}%

\usepackage{graphicx}
\usepackage{amsmath}
\usepackage{amsfonts}
\usepackage{amssymb}
\usepackage{accents}
\usepackage{enumerate}

\graphicspath{{figs/}}

\usepackage{array} 
\newcolumntype{R}[1]{>{\hbox to #1\bgroup\hfill$}c<{$\egroup}}

\usepackage{program}
\catcode`\|=12\relax

\newcommand{\norm}[1]{\left\lVert#1\right\rVert}
\newcommand{\bra}[1]{\left\langle#1\right|} 
\newcommand{\ket}[1]{\left|#1\right\rangle}

\begin{document}

\title[Renormalization method]{Renormalization method for proving frustration-free local spin chains are gapped}


\author[1]{\fnm{Ari} \sur{Mizel}}\email{ari@arimizel.com}

\author[2]{\fnm{Van} \sur{Molino}}\email{vjmoli2@super.org}

\affil[1]{
\orgname{Laboratory for Physical Sciences}, \orgaddress{\street{8050 Greenmead Drive}, \city{College Park}, \postcode{20740}, \state{MD}, \country{U.S.A}}}

\affil[2]{
\orgname{Center for Computing Sciences}, \orgaddress{\street{17100 Science Drive}, \city{Bowie}, \postcode{20715}, \state{MD}, \country{U.S.A}}}


\abstract{Key properties of a physical system depend on whether it is gapped, i.e. whether its spectral gap has a positive lower bound that is independent of system size.  In quantum information theory, the question of whether a system is gapped has essential computational significance as well.  Here, we introduce a rigorous  renormalization method to prove that a spin chain is gapped.  This approach exploits the fact that ground states of gapped systems exhibit decaying correlations.  We apply the method to show that two interesting models are gapped, successfully completing proofs even where the previously established methods are inconclusive.
}

\maketitle

\section{Introduction}
\label{intro}
The gap is the minimum amount of energy required to excite a quantum mechanical system. It is a critical determinant of the system's physical behavior.  For this reason, the field of condensed matter physics has long dichotomized ``metallic'' systems, for which the excitation energy tends to zero as the system grows larger, and "gapped" systems, for which the excitation energy tends to a positive constant.

As quantum information science has burgeoned, researchers have come to appreciate that the energy gap also has a vital computational significance.  It is possible to design a quantum mechanical system so that its ground state contains a desired computational result \cite{Kadowaki1998,Farhi01,Kitaev00,Mizel01}.  One can then try to bring the system into this ground state adiabatically \cite{Albash2018}.  The main parameter determining the runtime is the system's energy gap.

Moreover, classical computation of the ground state of a system tends to be much easier when the system is gapped.  This fact that has been recognized anecdotally for a long time, and recently a rigorous renormalization algorithm was proven to allow efficient computation of the ground state of gapped spin chains \cite{Landau2015,Block2020}.

Although it is of primary significance whether a system's Hamiltonian is gapped, resolving the mathematical question can be thorny.  The problem has been solved in the simplest case of translationally-invariant, frustration-free spin-$1/2$ chains \cite{Bravyi2015}.  However, in full generality, the question is computationally undecidable \cite{Cubitt2015} even in one dimension \cite{Bausch2020}.  In the restricted setting of frustration-free Hamiltonians with spin greater than 1/2, the local gap method \cite{Knabe1988} and the Martingale method \cite{Nachtergaele1996} constitute the main tools to show a system is gapped.  Unfortunately, they are often inconclusive.

In this paper, we leverage the local gap method with renormalization.  Renormalization is an operation that freezes out some of the degrees of freedom of the Hamiltonian, leaving a renormalized Hamiltonian of those that remain.  This process clearly discards information about the system.  However, we prove that, for frustration-free spin chains, the renormalized Hamiltonian is gapped if and only if the original Hamiltonian is gapped.  With this first result in hand, we can study the renormalized Hamiltonian to determine whether it is in fact gapped.  

According to a well-known rule of condensed matter physics, when there is an energy gap, the chain's ground state  $\ket{\Psi}$ exhibits spatially decaying correlations.  An energy gap is related to the correlation between variables $A$ and $B$ because $\bra{\Psi} A B \ket{\Psi} -\bra{\Psi} A \ket{\Psi} \bra{\Psi} B \ket{\Psi} = \bra{\Psi} A \left(\textsf{I} - \ket{\Psi}\bra{\Psi} \right) B \ket{\Psi}$,  and $\textsf{I} - \ket{\Psi}\bra{\Psi}$ is a projector on to the excited states of the chain.  Roughly speaking, for a gapped system, it is hard for spatially distant operators $A$ and $B$ to excite the ground state to the same excited state.
There are rigorous theorems proving that a gapped system must have decaying correlations \cite{Hastings06}, but the converse is not always true.  
Some additional assumptions are needed, beyond decaying correlations, to imply a system is gapped.  Here, we identify one such set of assumptions and use renormalization to prove a second result: a converse theorem for a one-dimensional spin chain.  Under renormalization, segments of a spin chain become effective renormalized spins.  Decaying ground state correlations then lead to neighboring effective couplings that nearly commute, a simplification that allows us to prove our converse theorem.


In the next section, Sec. 2, we present in detail the renormalization method for proving a spin chain is gapped.   We then showcase the power of the approach by applying it successfully to two interesting Hamiltonians on which the the local gap method and the Martingale method are inconclusive.
  The first such Hamiltonian, the ``teleportation'' chain \cite{Mizel21}, is studied in Sect. 3.  The second Hamiltonian, the ``swap'' chain, is introduced and analyzed in Sect. 4.

\section{Renormalization Method}

\subsection{Setup and Notation}
\label{2SS1} 

\begin{definition}
Consider a one-dimensional chain of spins.  The chain includes $\ell$ spins, each with $d$ states.   In addition, there is a boundary spin of dimension $d_0$ at one end and a boundary spin of dimension $d_{\ell+1}$ at the other.  Thus, the Hilbert space of the chain is $\mathbb{C}^{d_{\ell+1} }\otimes (\mathbb{C}^d)^{\otimes \ell} \otimes \mathbb{C}^{d_0}$.  The chain Hamiltonian is defined by
\begin{equation}
{\mathcal H} = \bar{H}^{\ell,\ell+1} \otimes \textsf{I}^{\otimes \ell}+\textsf{I} \otimes  \left( \sum_{i=0}^{\ell-2} \textsf{I}^{\otimes \ell-2-i} \otimes \bar{H} \otimes \textsf{I}^{\otimes i} \right) \otimes \textsf{I} +  \textsf{I}^{\otimes \ell} \otimes \bar{H}^{0,1}
 \label{mathcalH}
\end{equation}
where $\bar{H}: \mathbb{C}^d \otimes \mathbb{C}^d \rightarrow  \mathbb{C}^d \otimes \mathbb{C}^d$ and $\bar{H}^{0,1}: \mathbb{C}^d \otimes \mathbb{C}^{d_0} \rightarrow  \mathbb{C}^d \otimes  \mathbb{C}^{d_0}$  and $\bar{H}^{\ell,\ell+1}: \mathbb{C}^{d_{\ell+1}} \otimes \mathbb{C}^{d} \rightarrow   \mathbb{C}^{d_{\ell+1}} \otimes \mathbb{C}^{d}$ are positive semi-definite operators acting on adjacent spins.   Here, and throughout this paper, we adopt the convention that $M^{\otimes 0} \otimes B = 1 \otimes B = B$ and $B \otimes M^{\otimes 0}  =   B \otimes 1 = B$ for any operators $A$, $B$, and $M$.  Thus, the $i=0$ term  in the sum reduces to $\textsf{I}^{\otimes \ell-2} \otimes \bar{H}$, and the $i=\ell-2$ term reduces to $ \bar{H} \otimes \textsf{I}^{\otimes \ell-2}$.  We will solely consider frustration-free Hamiltonians for which $\mathrm{ker}\, {\mathcal H} \ne \emptyset$.  As the formal notation is a bit tedious, a simple depiction is provided as part of Fig. \ref{FancyHfigure}.
\end{definition}

In this paper, we will be interested in Hamiltonians with ground states that exhibit decaying correlations.  To study the consequences of this decay, it will be convenient to partition the chain into long segments such that the correlations between the ends of a segment are weak.   Every segment has length ${\bar{\ell}}$, and there are a total of $\lfloor \ell/{\bar{\ell}} \rfloor$ segments.  A remnant of size $\bar{\bar{\ell}} = (\ell - \lfloor \ell/{\bar{\ell}} \rfloor {\bar{\ell}})$ at the front of the chain is treated separately as described below in remark \ref{mathcalHremark}.   
\begin{definition}
Define $\bar{H}^S$ to be the Hamiltonian of a segment of length $\bar{\ell}$,
\[
\bar{H}^S = \sum_{i=0}^{\bar{\ell}-2} \textsf{I}^{\otimes \bar{\ell}-2-i} \otimes \bar{H} \otimes \textsf{I}^{\otimes i}.
\]
Since $\mathrm{ker}\, {\mathcal H} \ne \emptyset$, it follows that $\mathrm{ker}\, \bar{H}^S \ne \emptyset$.  Let $\bar{P}$ denote the projector on to its zero-energy ground states so that $ \bar{H}^S \bar{P} = \bar{P} \bar{H}^S = 0$.  Define 
\[
H^S =  \sum_{i = 0}^{\lfloor \ell/{\bar{\ell}} \rfloor - 1} \textsf{I}^{\otimes (\lfloor \ell/{\bar{\ell}} \rfloor -1 - i)\bar{\ell}} \otimes \bar{H}^S \otimes \textsf{I}^{i{\bar{\ell}}} 
\]
to be the sum of the Hamiltonians of the unlinked segments partitioning $\mathcal H$.  (The remnant of length $\bar{\bar{\ell}}$ has been omitted so that $H^S$ acts on $\lfloor \ell/{\bar{\ell}} \rfloor \bar{\ell}$ spins.)  Let $P^S = \bar{P}^{\otimes \lfloor \ell/{\bar{\ell}} \rfloor }$ denote the projector on to its zero-energy ground states so that $H^S P^S = P^S H^S = 0$.  Note that $ \textsf{I}^{\otimes \lfloor \ell/{\bar{\ell}} \rfloor  \bar{\ell} - i {\bar{\ell}}-1} \otimes \bar{H} \otimes \textsf{I}^{i {\bar{\ell}}-1}$ links the last spin $i {\bar{\ell}}$ of segment $i$ and the first spin $i {\bar{\ell}}+1$ of segment $i+1$.     Define the link Hamiltonian to be 
\[
H^L =  \sum_{i=1}^{\lfloor \ell/{\bar{\ell}} \rfloor -1} \textsf{I}^{\otimes \lfloor \ell/{\bar{\ell}} \rfloor  \bar{\ell} - i {\bar{\ell}}-1} \otimes \bar{H} \otimes \textsf{I}^{i {\bar{\ell}}-1};
\]
it links adjacent segments together.
Then we set
\begin{equation}
H = H^S + H^L
\label{H}
\end{equation}
to be the Hamiltonian of the linked segments.  See Fig. \ref{Hfigure} for an illustration.

\begin{figure}[H]
\begin{center}
\includegraphics[width=4.5in]{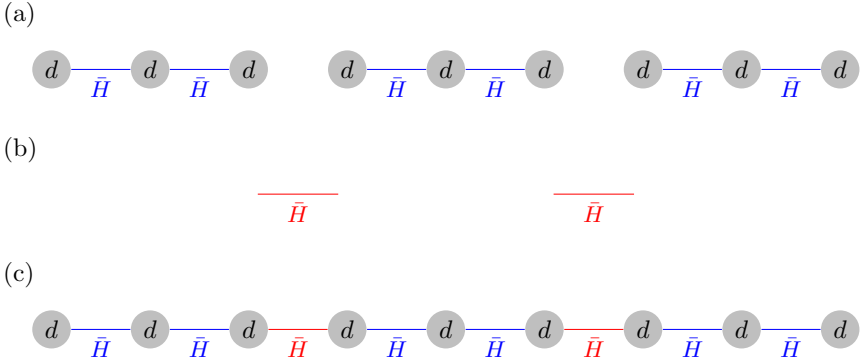}
\end{center}
\caption{Chain Hamiltonian $H$ for $\ell = 9$ spins of dimension $d$. (a) Hamiltonian $H^S$ of unlinked segments in the case $\bar{\ell} = 3$.  (b) Link Hamiltonian $H^L$.  (c) Hamiltonian $H = H^S+H^L$ of linked segments. }
\label{Hfigure}
\end{figure}

The Hamiltonian of the omitted remnant of length $\bar{\bar{\ell}}$ at the front of the chain is
\[
H^R =  \sum_{i= 0}^{\bar{\bar{\ell}}-2} \textsf{I}^{\otimes (\bar{\bar{\ell}}-2 - i) } \otimes \bar{H} \otimes \textsf{I}^{\otimes i}.
\]
\end{definition}

\begin{remark} 
It follows immediately from these definitions that
\begin{align} \nonumber
&\mathcal H =  \bar{H}^{\ell,\ell+1} \otimes \textsf{I}^{\otimes \ell} \\ \label{othermathcalH}
& \hspace{0.25in} + \textsf{I} \otimes (H \otimes  \textsf{I}^{\otimes \bar{\bar{\ell}}}+  \textsf{I}^{\otimes (\lfloor \ell/{\bar{\ell}} \rfloor \bar{\ell}-1)} \otimes \bar{H} \otimes \textsf{I}^{\otimes \bar{\bar{\ell}}-1} +\textsf{I}^{\otimes \lfloor \ell/{\bar{\ell}} \rfloor \bar{\ell}} \otimes  H^R) \otimes \textsf{I} \\ \nonumber
& \hspace{3.25in} +  \textsf{I}^{\otimes \ell} \otimes \bar{H}^{0,1},
\end{align}
as is depicted in Fig. \ref{FancyHfigure}.  We usually think of $\bar{\ell}$ as being large but fixed so the portion of the chain corresponding to $H$ is what tends to infinity along with system size.  Thus, the majority of the work involved in showing that ${\mathcal H}$ is gapped amounts to showing that $H$ is gapped.  We discuss this in more detail in a later remark.
\label{DecompRemark}
\end{remark}

\begin{figure}[H]

\begin{center}
\includegraphics[width=4.5in]{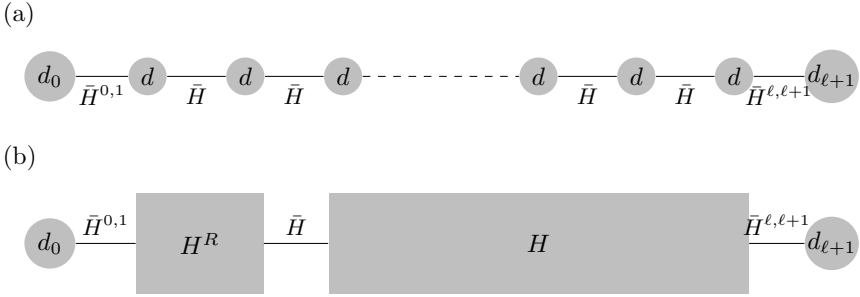}
\end{center}
\caption{(a) The Hamiltonian ${\mathcal H}$ corresponding to a spin chain with boundary conditions.  (b) The full Hamiltonian ${\mathcal H}$ partitioned into a convenient choice of pieces.}
\label{FancyHfigure}
\end{figure}

\begin{definition}
A primary object of our study will be $h = P^S H^L P^S$ in which states are constrained to have zero energy along the segments of $H$ and excitations occur only at the links.  Since $\bar{H}$ is the Hamiltonian coupling adjacent spins in $H$, we label the analogue for $h$ using the symbol $\bar{h}$.  Let $\bar{d}$ denote the number of ground states of a segment of length ${\bar{\ell}}$, i.e. let $\bar{d}$ be the dimension of the image of $\bar{P}$.   Define
\begin{equation}
\bar{h} = \left(\bar{P} \otimes \bar{P}\right) \left( \textsf{I}^{\otimes {\bar{\ell}}-1} \otimes \bar{H} \otimes \textsf{I}^{\otimes {\bar{\ell}}-1}\right)  \left(\bar{P} \otimes \bar{P}\right).
\end{equation}
In terms of $\bar{h}$, we can write
\begin{equation}
h =  \sum_{i=0}^{\lfloor \ell/{\bar{\ell}} \rfloor -2} \bar{P}^{\otimes \lfloor \ell/{\bar{\ell}} \rfloor -2 - i}  \otimes \bar{h}   \otimes  \bar{P}^{\otimes i}  
\label{h}
\end{equation}
\end{definition}
An illustration of $h$ appears in Fig. \ref{barhfigure}

\begin{figure}[H]
\begin{center}
\includegraphics[width=4.5in]{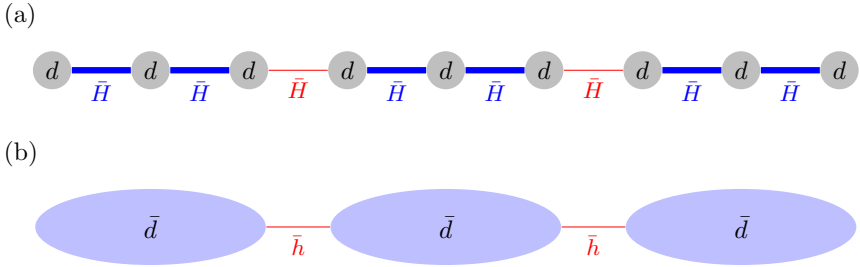}
\end{center}

\caption{(a) Three segments in the middle of a chain (remainder of chain not depicted).  Thick lines connect the spins within each segment (${\bar{\ell}} = 3$ in this figure).  These couplings are locked at zero energy for states in the image of $P^S$.  Thin lines denote links between segments.  (b) Renormalized version of (a).  Each segment is now an effective spin of dimension $\bar{d}$.  Links between effective spins are now the coupling $\bar{h}$.}
\label{barhfigure}
\end{figure}

\subsection{Gap at Some Scale Implies Gap at All Scales}
\label{2SS2} 

One sees that $h = P^S H^L P^S$ is a renormalized version of $H$ with $\bar{h}$ coupling effective spins of dimension $\bar{d}$.  Our initial results relate $h$ to $H$, and then we investigate $h$ itself assuming decaying correlations in the ground state.
Our argument will make use of the following technical lemma \cite{Mizel07}.  In the following lemma and throughout the paper, for any Hermitian matrix $A$ of dimension $n$, denote the ordered eigenvalues as $\lambda_0(A) \le \lambda_1(A) \le \dots \le \lambda_{n-1}(A)$. 
\begin{lemma}
Consider Hermitian matrices $A = B + C$ with $C \ge 0$.   Define the projector $P^B$
on to the 
space of eigenvectors of $B$ with eigenvalue $\lambda_0(B)$.   Let $c_0 < c_1$ denote the lowest two eigenvalues of
$P^BCP^B$ within the image of $P^B$.  Suppose $c_0 = 0$.  Then the spectral gap of $A$ satisfies 
\[
\lambda_1(A) - \lambda_0(A) \ge \frac{\lambda_1(B)-\lambda_0(B) }{\lambda_1(B)-\lambda_0(B)+ c_1+\norm{C}_2} c_1.
\]
\label{ABC}
\end{lemma}
\begin{proof}
Set $\left| \lambda_1(A) \right\rangle =  \alpha \left|\mu \right\rangle + \beta \left|\nu \right\rangle$ for $\alpha = \norm{P^B \left| \lambda_1(A) \right\rangle}$, $\beta = \norm{(\textsf{I}-P^B) \left| \lambda_1(A) \right\rangle}$, $\left|\mu \right\rangle = P^B \left| \lambda_1(A) \right\rangle/\norm{P^B \left| \lambda_1(A) \right\rangle}$,  and $\left|\nu \right\rangle = (\textsf{I}-P^B) \left| \lambda_1(A) \right\rangle/\norm{(\textsf{I}-P^B) \left|\lambda_1(A) \right\rangle}$.  The equation $(A-\lambda_0(A)) \left| \lambda_1(A) \right\rangle =(\lambda_1(A)-\lambda_0(A)) \left| \lambda_1(A) \right\rangle$ implies the $2 \times 2$ matrix equation
\begin{align*}
\left[\begin{array}{cc} \bra{\mu}A-\lambda_0(A)\ket{\mu} & \bra{\mu}A-\lambda_0(A)\ket{\nu} \\ \bra{\nu}A-\lambda_0(A)\ket{\mu} & \bra{\nu}A-\lambda_0(A)\ket{\nu} \end{array} \right] \left[\begin{array}{c}\alpha \\ \beta \end{array} \right] = (\lambda_1(A)-\lambda_0(A)) \left[ \begin{array}{c}\alpha \\ \beta \end{array}\right] \!\!.
\end{align*}
The smallest $\lambda_1(A)-\lambda_0(A)$ that can possibly satisfy this equation is the smaller eigenvalue of the matrix on the left hand side.  Since $A-\lambda_0(A) \textsf{I}$ is positive semi-definite, this smaller eigenvalue is bounded below by the determinant of the matrix  divided by its trace (i.e. the product of the matrix's eigenvalues divided by the sum of its eigenvalues).  Thus,
\begin{align*}
\lambda_1(A)-\lambda_0(A) \ge & \frac{\bra{\mu}A-\lambda_0(A)\ket{\mu} \bra{\nu}A-\lambda_0(A)\ket{\nu} - \left|\bra{\mu}A-\lambda_0(A)\ket{\nu}\right|^2}{\bra{\mu}A-\lambda_0(A)\ket{\mu}+ \bra{\nu}A-\lambda_0(A)\ket{\nu}}.
\end{align*}
To simplify the right hand side, we show $\lambda_0(A) = \lambda_0(B)$.  Clearly, $\lambda_0(B) \le \lambda_0(A)$ since $A=B+C$ with $C \ge 0$.  Let $\ket{c_0}$ by any vector satisfying $P^B\ket{c_0} = \ket{c_0}$ and  $P^BCP^B\ket{c_0} =c_0 \ket{c_0}=0$.  We have $0 = \bra{c_0} P^B C P^B \ket{c_0}/\bra{c_0}P^B\ket{c_0}= \bra{c_0} C \ket{c_0}/\left\langle c_0 | c_0\right\rangle$, so $\ket{c_0}$ is an eigenvector of $C$ with eigenvalue $0$.  Moreover, $B \ket{c_0} =  B P^B  \ket{c_0} = \lambda_0(B)  \ket{c_0}$.  It follows that $A\ket{c_0} = \lambda_0(B)\ket{c_0}$, implying $\lambda_0(A) \le \lambda_0(B)$ since $\lambda_0(B)$ is an eigenvalue of $A$ and $\lambda_0(A)$ is the smallest eigenvalue of $A$.  We conclude that $\lambda_0(A)=\lambda_0(B)$ and write
\begin{align*}
\lambda_1(A)-\lambda_0(A) \ge & \frac{\bra{\mu}C\ket{\mu} \bra{\nu}B+C-\lambda_0(A)\ket{\nu} - \left|\bra{\mu}C\ket{\nu}\right|^2}{\bra{\mu}C\ket{\mu}+ \bra{\nu}B-\lambda_0(A)\ket{\nu}+ \bra{\nu}C\ket{\nu}}
\end{align*}
because $B \ket{\mu} =BP^B \ket{\mu} = \lambda_0(B) \ket{\mu}= \lambda_0(A) \ket{\mu}$.  Simplifying further,
\[
\lambda_1(A)-\lambda_0(A)  \ge \frac{\bra{\mu}C \ket{\mu} \bra{\nu} B-\lambda_0(B)\ket{\nu}}{\bra{\mu}C\ket{\mu}+ \bra{\nu}B-\lambda_0(B)\ket{\nu}+ \norm{C}_2} 
\]
since $\bra{\mu}C\ket{\mu} \bra{\nu}C\ket{\nu}- \left|\bra{\mu}C\ket{\nu}\right|^2 \ge 0$ for the positive semi-definite matrix $C$.  

As shown above, any $\ket{c_0}$ satisfying $P^B\ket{c_0} = \ket{c_0}$ and  $P^BCP^B\ket{c_0} =c_0 \ket{c_0}=0$ must be an eigenvector of $A$ with eigenvalue $\lambda_0(A) \ne \lambda_1(A)$.  It must therefore be orthogonal to $\ket{\lambda_1(A)}$ and also to $P^B\ket{\lambda_1(A)}$, since $0=\bra{\lambda_1(A)}\left.c_0\right\rangle = \bra{\lambda_1(A)}P^B\ket{c_0}$.  It follows that
\[
c_1 \le \frac{\bra{\lambda_1(A)}P^BCP^B\ket{\lambda_1(A)}}{\bra{\lambda_1(A)}P^B\ket{\lambda_1(A)}}  = \bra{\mu}C \ket{\mu}.
\]
Inserting this into our last inequality for $\lambda_1(A)-\lambda_0(A)$ completes the proof.
 \end{proof}
 The following result of Lemm and Mosgunov \cite{Lemm2019} will also play a central role in our analysis.
\begin{lemma}
Consider a frustration-free Hamiltonian $H$ as in (\ref{H}). If $H$ is not gapped, then there exists a strictly increasing sequence of chain sizes $\ell_i$ for $i \in \mathbb{N}$ such that the gap at size $\ell_i$ is less than or equal to $4 \sqrt{6} \ell_i^{-3/2}$.
\label{LemmMosgunov}
\end{lemma}
\begin{proof}
A proof by contradiction follows.  Suppose the statement is not true: $H$ is not gapped but there nevertheless exists some positive constant $l$ such that its gap $\gamma_{l^\prime} > 4 \sqrt{6}  {l^\prime}^{-3/2}$ for all ${l^\prime}\ge l$.  We write $\gamma_{l^\prime} = 4 \sqrt{6}  {l^\prime}^{-3/2} + \Delta \gamma_{l^\prime}$ where $\Delta \gamma_{l^\prime} > 0$ for all ${l^\prime}\ge l$. Then we use the corollary of theorem 2.7 described in remark 2.8 (ii) of \cite{Lemm2019}.  The corollary states 
\[
\gamma_m \ge  \frac{1}{2^8\sqrt{6(2l)}} (\min_{l \le l^\prime \le 2l} \gamma_{l^\prime} -4\sqrt{6}(2l)^{-3/2})
\]
where we have set $n=2l$.  This allows us to conclude that 
\begin{align*}
\gamma_m &=\frac{1}{2^8\sqrt{6(2l)}}(\min_{l \le l^\prime \le 2l} (4\sqrt{6} {l^\prime}^{-3/2} + \Delta \gamma_{l^\prime}) -4\sqrt{6}(2l)^{-3/2}) \\
&\ge \frac{1}{2^8\sqrt{6(2l)}} \min_{l \le l^\prime \le 2l}  \Delta \gamma_{l^\prime}
\end{align*}
which shows that $H$ is gapped, drawing a contradiction.
\end{proof}

\begin{theorem}
$H$ is gapped if and only if $h$ is gapped.
\label{hH}
\end{theorem}
 
\begin{proof}
First, we show that, if $h$ is not gapped, then $H$ is not gapped.  Because $H$ is a frustration-free Hamiltonian, the zero energy ground state of $H$ lies in the image of $P^S$ and is also a zero energy ground state of $h$.  (Of course, any state in the nullspace of $P^S$ is an additional zero energy, trivial eigenstate of $h$.)  Suppose that $h$ has a low energy excited eigenstate $\ket{\lambda}$.  Because $\ket{\lambda}$ is orthogonal to the ground state of $h$, it is orthogonal to the ground state of $H$ and constitutes a variational guess for the first excited state of $H$.  Thus, the energy $\bra{\lambda} h \ket{\lambda}= \bra{\lambda} H \ket{\lambda}$ is an upper bound for the spectral gap of $H$.

Next, we show that, if $h$ is gapped then $H$ is gapped.  Set $A = H$, $B = H^S$, and $C = H^L$ in lemma \ref{ABC}.  The parameter $c_1$ equals the gap of $h$.   The parameter $\lambda_1(B) - \lambda_0(B)$ equals the gap of $H^S$, which is a positive value depending upon the segment length ${\bar{\ell}}$ but not upon the chain length of $H$.   Since $\norm{H^L}_2 \le \lfloor \ell/{\bar{\ell}} \rfloor \norm{\bar{H}}_2$, lemma \ref{ABC} implies that the gap of $H$ decreases no faster than $\mathrm{const}/(\mathrm{const}^\prime+\lfloor \ell/{\bar{\ell}} \rfloor )$.  Thus, there exists a positive integer $n$ such that the gap of $H$ is greater than $4\sqrt{6}(\lfloor \ell/{\bar{\ell}} \rfloor  )^{-3/2}$ for all $\lfloor \ell/{\bar{\ell}} \rfloor > n$.
According to lemma 2, this implies that $H$ must be gapped.
 \end{proof}
 
\begin{remark}
The theorem applies to $H$ with open boundary conditions in which there are no boundary spins to complicate the  renormalization of $H$ to $h$.  One expects to be able to frame an analogous theorem when $H$ has periodic boundary conditions, since one can prove lemma \ref{LemmMosgunov} by invoking theorem 3 appearing in Gosset and Mozgunov \cite{Gosset2016} instead of invoking the the corollary of theorem 2.7 of Lemm and Mozgunov  \cite{Lemm2019}.
\end{remark}

\begin{remark}
The theorem can be used in conjunction with lemma \ref{ABC} to prove that ${\mathcal H}$ is gapped, including the remnant part of the chain and the boundary spins.  To address the remnant part, set $A = B+C$ where $B =  H \otimes \textsf{I}^{\otimes \bar{\bar{\ell}}}+ \textsf{I}^{\otimes \lfloor \ell/{\bar{\ell}} \rfloor \bar{\ell}} \otimes  H^R$, and $C = \textsf{I}^{\lfloor \ell/{\bar{\ell}} \rfloor \bar{\ell}-1 } \otimes \bar{H} \otimes \textsf{I}^{\otimes \bar{\bar{\ell}}-1} $.  If $h$ is gapped then $H$ is gapped so $B$ is gapped.  After all, the gap of $H^R$ is bounded from below by a positive constant independent of $\ell$: we can choose the constant to be the minimum of its gap over the length $0 \le \bar{\bar{\ell}} < \bar{\ell}$.  Clearly, $\norm{C}_2$ is independent of chain size.  So, $A$ is gapped provided $h$ is gapped and the gap $c_1$ of $C$ in the basis of ground states of $B$ is independent of chain size.  Once one proves this for a given spin chain with a specific $\bar{H}$, lemma \ref{ABC} implies that $H \otimes \textsf{I}^{\otimes \bar{\bar{\ell}}}+ \textsf{I}^{\lfloor \ell/{\bar{\ell}} \rfloor \bar{\ell}-1 } \otimes \bar{H} \otimes \textsf{I}^{\otimes \bar{\bar{\ell}}-1}  +\textsf{I}^{\otimes \lfloor \ell/{\bar{\ell}} \rfloor \bar{\ell}} \otimes  H^R$ is gapped.  To consider the boundary terms, we can apply lemma \ref{ABC} a second time with 
\begin{align*}
& B = \textsf{I} \otimes \left(H \otimes \textsf{I}^{\otimes \bar{\bar{\ell}}}+ \textsf{I}^{\lfloor \ell/{\bar{\ell}} \rfloor \bar{\ell}-1 } \otimes \bar{H} \otimes \textsf{I}^{\otimes \bar{\bar{\ell}}-1}  +\textsf{I}^{\otimes \lfloor \ell/{\bar{\ell}} \rfloor \bar{\ell}} \otimes  H^R \right)  \otimes \textsf{I}, \\
& C =  \bar{H}^{\ell,\ell+1} \otimes \textsf{I}^{\otimes \ell}+\textsf{I}^{\otimes \ell} \otimes \bar{H}^{0,1} \\
& A = B+C = \mathcal{H}.
\end{align*}
Of course, for some forms of ${\mathcal H}$, it may be more convenient to incorporate the effect or a boundary term in $H$ or in $H^R$ as a first step.  Then, as a second step, lemma \ref{ABC} can be invoked using $C = \textsf{I} \otimes (\textsf{I}^{\lfloor \ell/{\bar{\ell}} \rfloor \bar{\ell}-1 } \otimes \bar{H} \otimes \textsf{I}^{\otimes \bar{\bar{\ell}}-1}) \otimes \textsf{I}$ to couple the remnant to the rest of the chain.
\label{mathcalHremark}
\end{remark}
 
\subsection{Decaying Correlations and the Gap of the Renormalized Hamiltonian}
\label{2SS3} 
 
Having related the renormalized Hamiltonian $h$ to the total chain Hamiltonian ${\mathcal H}$, we now study $h$ itself.  The following definition describes decaying ground state correlations in terms of the properties of $h$.
\begin{definition} 
Write $\bar{h} = \tilde{h} + \bar{k}$  where $\tilde{h} \ge 0$ and $[\bar{P} \otimes \tilde{h},\tilde{h} \otimes \bar{P}] = 0$.  If, for any $\Delta>0$, a segment length ${\bar{\ell}}$ can be found such that $\norm{\bar{k}}_2 \le \Delta$, then the chain has decaying correlations.
\label{decayingcorrelations}
\end{definition}
  
\begin{remark}
Consider an alternate way to define decaying ground state correlations in a set of ground states $\ket{\Psi_z}$ with projector $\sum_{z^\prime} \ket{\Psi_{z^\prime}} \bra{\Psi_{z^\prime}}$.  One might stipulate that, given any $\Delta >0$, it is possible to find a distance ${\bar{\ell}}$ such that any $A$ and $B$ acting on parts of the chain spaced ${\bar{\ell}}$ or more apart satisfy
\[
\left| \bra{\Psi_z} A B \ket{\Psi_z} - \bra{\Psi_z} A \left(\sum_{z^\prime} \ket{\Psi_{z^\prime}} \bra{\Psi_{z^\prime}}\right) B \ket{\Psi_z} \right| < \Delta \norm{A} \norm{B}.
\]
For brevity, we write this equation $\bra{\Psi_z} A \left(\sum_{z^\prime} \ket{\Psi_{z^\prime}} \bra{\Psi_{z^\prime}}\right) B \ket{\Psi_z}  \approx \bra{\Psi_z} A B \ket{\Psi_z} $.  This alternate definition motivates Def. \ref{decayingcorrelations} as follows.  Consider a chain with Hamiltonian $H^S$, without the link Hamiltonian $H^L$.  Then $\sum_{z^\prime} \ket{\Psi_{z^\prime}} \bra{\Psi_{z^\prime}}$ is the projector $P^S$ on the ground state of $H^S$.  Set $A = \textsf{I}^{\otimes \lfloor \ell/{\bar{\ell}} \rfloor \bar{\ell}- i {\bar{\ell}}-1} \otimes \bar{H} \otimes \textsf{I}^{i {\bar{\ell}}-1}$ and $B = \textsf{I}^{\otimes \lfloor \ell/{\bar{\ell}} \rfloor \bar{\ell}- (i+1) {\bar{\ell}}-1} \otimes \bar{H} \otimes \textsf{I}^{(i+1) {\bar{\ell}}-1}$.  Then $ \bra{\Psi_z}A P^S B\ket{\Psi_z} \approx \bra{\Psi_z}A B\ket{\Psi_z} = \bra{\Psi_z}B A\ket{\Psi_z}  \approx \bra{\Psi_z}B P^S A\ket{\Psi_z}$, implying that   $\bra{\Psi_z} [A P^S B  ,B P^S A ] \ket{\Psi_z} \approx 0$.  This means that 
\[
 \bra{\Psi_z} \bar{P}^{\otimes \lfloor \ell/{\bar{\ell}} \rfloor - (i+1) - 3} \otimes [\bar{P} \otimes \bar{h},\bar{h} \otimes \bar{P}] \otimes \bar{P}^{\otimes i} \ket{\Psi_z}\approx 0 
\]
which leads to $[\bar{P} \otimes \bar{h},\bar{h} \otimes \bar{P}] \approx 0$ as required by Def. \ref{decayingcorrelations}.
\end{remark}

Def. \ref{decayingcorrelations} implies that ${\bar{\ell}}$ can be chosen to render  $h$ close to a sum of nearly commuting Hamiltonians.  As a result, the gap of $h$ should approximately equal the gap of the $\bar{h}$ up to a small correction.  The proof employs the following standard argument \cite{Knabe1988}.  In the following proposition and throughout the paper, let $\{A,B\} = AB + BA$ denote the anticommutator.

\begin{proposition}
Let  $\bar{g} > 0$ be the gap of $\bar{h}$.  If 
\[ \left\{\bar{P} \otimes \bar{h},\bar{h} \otimes \bar{P}\right\} + x \left(\bar{P} \otimes \bar{h}+\bar{h} \otimes \bar{P}\right) \ge 0
\] then the gap $g$ of the renormalized Hamiltonian (\ref{h}) satisfies $g  \ge \bar{g} - 2x$.
\end{proposition}
\begin{proof}
Divide $h$ by $\bar{g}$ and square to obtain the inequality
\begin{align*}
\frac{h^2}{\bar{g}^2} & \ge    \sum_{i=0}^{\lfloor \ell/{\bar{\ell}} \rfloor -2}\bar{P}^{\otimes \lfloor \ell/{\bar{\ell}} \rfloor -2 - i}  \otimes \frac{\bar{h}^2}{\bar{g}^2}   \otimes  \bar{P}^{\otimes i} \\
& + \sum_{i=0}^{\lfloor \ell/{\bar{\ell}} \rfloor -3}\bar{P}^{\otimes \lfloor \ell/{\bar{\ell}} \rfloor -3 - i}  \otimes \left\{ \bar{P} \otimes \frac{\bar{h}}{\bar{g}},\frac{\bar{h}}{\bar{g}} \otimes \bar{P}  \right\}   \otimes  \bar{P}^{\otimes i} \ge \frac{h}{\bar{g}} - \frac{2x h}{\bar{g}^2} = \frac{1}{\bar{g}^2} (\bar{g} - 2 x ) h
\end{align*}
implying a gap of $g  \ge \bar{g} - 2x$ for $h$. 
\end{proof}

Thus, if $x$ is small, the gap $g$ of $h$ approaches $\bar{g}$ as expected.   The following proposition gives conditions that allow a suitable $x$ to be found that indeed becomes small as $\Delta$ becomes small.

As stated above, for any Hermitian matrix $A$ of dimension $n$, we denote the ordered eigenvalues as $\lambda_0(A) \le \lambda_2(A) \le \dots \le \lambda_n(A)$. The eigenvalue stability inequality states that $\left|\lambda_i(A+B) - \lambda_i(A)\right| \le \norm{B}_2$ for any $0\le i \le n-1$.

\begin{proposition}  Consider a spin chain with decaying correlations as in Def. \ref{decayingcorrelations} with $\norm{\bar{k}}_2 \le \Delta$.  Set $\eta \ge \lVert \tilde{h} \rVert_2$.  Let $z$ be the dimension of the kernel of $\bar{h} \otimes \bar{P} + \bar{P} \otimes \bar{h}$.  Suppose $\lambda_{z}(\tilde{h} \otimes \bar{P} + \bar{P} \otimes \tilde{h}) \ge \tilde{g} > 2 \Delta$ .  If
\[
x > \frac{2 \Delta (2 \eta + \Delta)}{\tilde{g} - 2\Delta}
\]
then $\{\bar{h} \otimes \bar{P},\bar{P} \otimes \bar{h} \} + x (\bar{h} \otimes \bar{P} + \bar{P} \otimes \bar{h}) \ge 0$.
\end{proposition}
\begin{proof}
We use the eigenvalue stability inequality to prove the proposition by contradiction. Write $\{\bar{h} \otimes \bar{P},\bar{P} \otimes \bar{h} \} + x (\bar{h} \otimes \bar{P} + \bar{P} \otimes \bar{h}) = A+B$ where $A=\{\tilde{h} \otimes \bar{P},\bar{P} \otimes \tilde{h} \} + x (\tilde{h} \otimes \bar{P} + \bar{P} \otimes \tilde{h})$ and $B = \{\bar{k} \otimes \bar{P}, \bar{P} \otimes \tilde{h}\} +  \{\tilde{h} \otimes \bar{P}, \bar{P} \otimes \bar{k}\} +  \{\bar{k} \otimes \bar{P}, \bar{P} \otimes \bar{k}\} + x (\bar{k} \otimes \bar{P}+ \bar{P} \otimes \bar{k})$.  Suppose that $A + B$ has a negative eigenvalue.  Because $\bar{h} \ge 0$, every vector in the kernel of $\bar{h} \otimes \bar{P} + \bar{P} \otimes \bar{h}$ is also in the kernel of $\bar{P} \otimes \bar{h}$, in the kernel of $\bar{h} \otimes \bar{P}$, and thus in the kernel of $A+B$.   It follows that  $\lambda_{z}(A+B)\le 0$.

From the form of $A$, and our assumption $\lambda_{z}(\tilde{h} \otimes \bar{P} + \bar{P} \otimes \tilde{h}) \ge \tilde{g}$, we see that  $\lambda_{z}(A) \ge x \tilde{g}$. Thus,
\[
x \tilde{g} \le \left| \lambda_{z}(A+B) - \lambda_{z}(A)\right| \le \norm{B}_2 \le 4\Delta \eta + 2 \Delta \Delta + 2 x \Delta.
\]
Subtract $2 x \Delta$ from both sides and then divide through by $\tilde{g}-2 \Delta$ to obtain
\begin{equation}
x \le \frac{2 \Delta(2 \eta + \Delta)}{\tilde{g} - 2 \Delta}.
\label{xlessthan}
\end{equation}
The proposition follows.
\end{proof}

Assembling our results, we arrive at the following theorem.
\begin{theorem}
Consider a chain exhibiting decaying ground state correlations as defined in Def \ref{decayingcorrelations}.  Let $z$ be the dimension of the kernel of $\bar{h} \otimes \bar{P} + \bar{P} \otimes \bar{h}$.   Define $\tilde{g}$ such that $\tilde{g} \le \lambda_{z}(\tilde{h} \otimes \bar{P} +  \bar{P} \otimes \tilde{h})$.  Let the gap of $\bar{h}$ be $\bar{g}$.  Suppose, for sufficiently large ${\bar{\ell}}$, both $g$ and $\tilde{g}$ are greater than positive, ${\bar{\ell}}$-independent constants. Then $H$ is gapped.
\label{Hgapped}
\end{theorem}
 \begin{proof}
By choosing ${\bar{\ell}}$ sufficiently large, and thus $\Delta$ sufficiently small, we can ensure that $\tilde{g}> 2 \Delta$.
Using propositions 1 and 2, we see that
\[
\bar{g} - \frac{4 \Delta (2 \eta + \Delta)}{\tilde{g}  - 2\Delta}
\]
is a lower bound of the gap of $h$.  Theorem \ref{hH} then completes the proof.
\end{proof}

\begin{remark}
Example 2 of \cite{Nachtergaele1996} provides an instance of a chain with decaying ground state correlations that is not gapped.  Theorem \ref{Hgapped} does not apply to this chain because its $\bar{g}$ falls polynomially to zero with $\bar{\ell}$.  This is a remarkable property: by exciting one bond $\bar{H}$ that couples two specific spins, it is possible to form a variational excited state whose energy goes to zero with chain length.  To say this in another way, by appropriately glueing a ground state on half of the chain together with a ground state on the other half of the chain, one can form an excited state whose energy goes to zero with chain length.
\end{remark}

\begin{remark}
Equation (\ref{h}) and the expression $\bar{h} = \tilde{h} + \bar{k}$ in definition \ref{decayingcorrelations} show that the renormalized Hamiltonian $h$ is nearly the sum of $\tilde{h}$ terms that commute.  Thus, one might expect a lower bound on the gap of $h$ given by the gap of $\tilde{h}$ minus $\norm{\bar{k}}_2$ or the gap of $\bar{h}$ minus $\norm{\bar{k}}_2$.  However, a subtlety arises since $\bar{h}$ typically has a kernel of dimension $z > 1$.  What if $\tilde{h}=\bar{h} -\bar{k}$ exhibits a small gap between the elements of the kernel $\bar{h}$, leading to a small gap for $h$?  The technical conditions of proposition 2 and theorem 2 ensure this does not happen.  And instead of a lower bound on the gap of $h$ given simply by the gap of $\tilde{h}$ minus $\norm{\bar{k}}_2$, we find instead the gap of $\tilde{h}$ minus a prefactor times $\norm{\bar{k}}_2$.
\end{remark}

\subsection{Using the Renormalization Method in Practice}
\label{2SS4} 

Theorem \ref{Hgapped} can be applied in a variety of settings in order to establish gaps, even when previously known techniques do not immediately apply. We will showcase two such examples, both of which are motivated by quantum circuits.  While the details of rigorously proving that these systems are gapped vary slightly, the high-level strategy remains the same.  We provide a summary here in order to foreshadow what is to come. In some sense, one can view this as the following procedure:
\[
\text{computing ground state space} \rightsquigarrow \text{proving a gap (if one exists)},
\]
which serves as a sort of loose converse to the results of \cite{Landau2015,Block2020}.
\begin{enumerate}[(a)]

    \item In order to apply Theorem \ref{Hgapped} we must analyze the renormalized Hamiltonian.  This involves computing $\bar{P}$, which amounts to understanding the ground state space of $\bar{H}^S$, an open chain of length $\bar{\ell}$. Induction can often be used here.  Assuming the form of the ground state space of a chain of length $\bar{\ell}$, we establish that of a chain of length $\bar{\ell}+1$. In light of frustration-freeness, this reduces to finding the kernel of the operator $\hat{h}$ corresponding to the following bond:

    \begin{figure}[H]
    \begin{center}
    \includegraphics[width=3.5in]{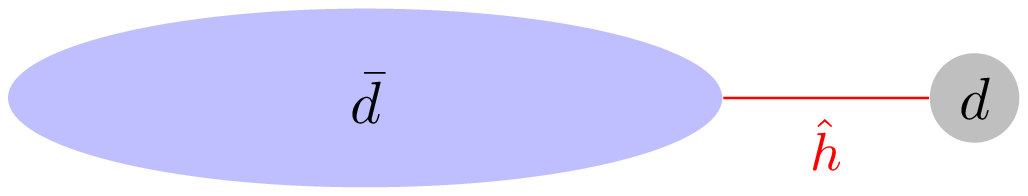}
    \end{center}
    \label{bond}
    \end{figure}
    
    In other words, we can restrict the leftmost $\bar{\ell}$ spins to live in the ground state space of their open chain while the single new spin is treated as free.
   
    \item Now that we know $\bar{P}$, we can compute $\bar{h}$, which is the main object driving the renormalized Hamiltonian. Recall that $\bar{h}$ corresponds to the following bond:
   
    \begin{figure}[H]
  \begin{center}
\includegraphics[width=4.5in]{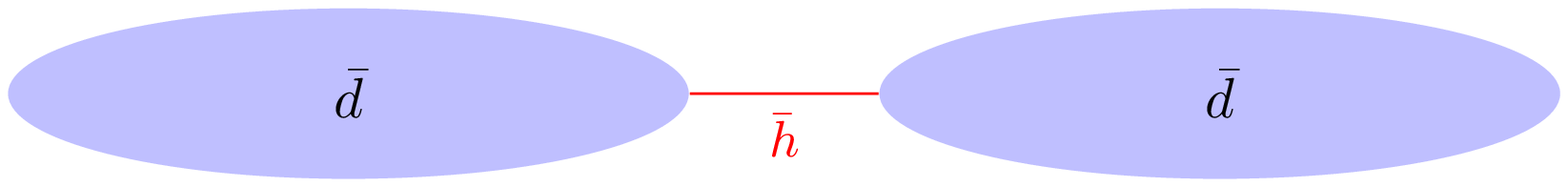}
\end{center}
    \label{barhbond}
    \end{figure}
    
    We note similarities with $\hat{h}$ and often find that the same arguments apply.
    
    \item Next, we show decaying correlations by finding an operator $\tilde{h}$ that is sufficiently close to $\bar{h}$ and satisfies $\left[ \bar{P} \otimes \tilde{h} , \tilde{h} \otimes \bar{P} \right] = 0$.  This step can be a bit subtle, as we find that different forms of $\tilde{h}$ may be appropriate for different flavors of $\bar{h}$.
    
    \item Then, we check the technical conditions associated with lower bounding the gaps $\bar{g}$ and $\tilde{g}$, as required by Theorem \ref{Hgapped}.

    \item Finally, we apply Remark \ref{mathcalHremark} in order to deal with the remnant part of the chain and the boundary conditions, allowing us to conclude that the original Hamiltonian ${\mathcal H}$ is gapped.
    
\end{enumerate}

\section{Teleportation chain Hamiltonian}

\subsection{Description of and motivation for the teleportation chain}
\label{3SS1} 

We now apply this framework to prove that the one-dimensional teleportation chain Hamiltonian \cite{Mizel21} is gapped.   In this definition and in the rest of the paper, we drop the symbol $\otimes$ from tensor products of states when no confusion will arise.
\begin{definition}
Consider a chain consisting of a one-dimensional array of $2\ell$ sites, each inhabited by a space $\mathbb{C}^3$ (a qutrit), with a boundary site inhabited by a space $\mathbb{C}^{2}$ (a qubit) on each end. Label the states of each qutrit $\left\{\ket{0},\ket{1},\ket{\scriptscriptstyle IDLE}\right\}$ and of each qubit by $\left\{\ket{0},\ket{1}\right\}$.
Let the teleportation chain Hamiltonian be comprised of terms of the form
\begin{align*}
{H}^{B}  =  (\ket{0} \otimes \ket{1} )( \bra{0} \otimes \bra{1})+(\ket{1} \otimes \ket{0} )(\bra{1} \otimes \bra{0}) \\
+ \frac{1}{2}(\left|0\right> \otimes \left|0\right>-\left|1\right> \otimes \left|1\right>)(\left<0\right| \otimes \left<0\right|-\left<1\right| \otimes \left<1\right|)
\end{align*}
and terms of the form
\begin{align*}
 {H}^{P}(\theta) =   \sum_{b=0,1} (\left|{\scriptscriptstyle IDLE}\right>\otimes  \left|b\right>)(  \left<{\scriptscriptstyle IDLE}\right| \otimes \left<b\right| ) +( \left|b\right> \otimes  \left|{\scriptscriptstyle IDLE}\right>)(\left<b\right| \otimes \left<{\scriptscriptstyle IDLE}\right| ) + \nonumber\\
\left(\sin \theta  \frac{\left|0\right> \otimes \left|0\right>+\left|1\right> \otimes \left|1\right>}{\sqrt{2}}  -\cos \theta \left|{\scriptscriptstyle IDLE}\right>\otimes \left|{\scriptscriptstyle IDLE}\right> \right) \hspace{1.0in} \\
 \left( \sin \theta \frac{\left<0\right| \otimes \left<0\right|+\left<1\right| \otimes \left<1\right|}{\sqrt{2}}- \cos \theta \left<{\scriptscriptstyle IDLE}\right| \otimes \left<{\scriptscriptstyle IDLE}\right|\right) \nonumber
\end{align*}
alternating.  The ${H}^{P}(\theta)$ terms act between qutrits $2i-1$ and $2 i$, for $i = 1$ to $\ell$, effecting a projection from a Bell state on to $\left|{\scriptscriptstyle IDLE}\right> \otimes \left|{\scriptscriptstyle IDLE}\right> $.  The ${H}^{B}$ terms create Bell pairs between qutrits $2i$ and $2i+1$, for $i = 1$ to $\ell-1$.  At the ends of the chain, there is a term ${H}^{B}$ acting between boundary qubit 0 and qutrit $1$ and a term ${H}^{B}$ acting between qutrit $2\ell$ and boundary qubit $2\ell+1$.  (Strictly speaking, an $H^{B}$ operator acting on a qutrit and a boundary qubit is not the same as an $H^{B}$ operator acting on two qutrits because their domains have different dimensions, but we will abuse notation and ignore this distinction.)  

\begin{figure}[H]
\begin{center}
\includegraphics[width=4.5in]{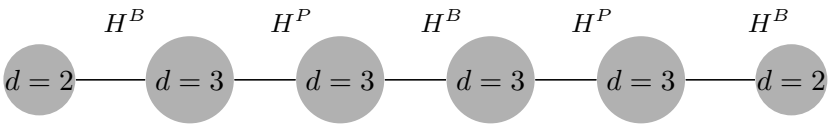}
\end{center}

\caption{Teleportation chain with $2\ell =4$ qutrits and a boundary qubit at each end.}
\label{teleportationchain}
\end{figure}

\end{definition}

Fig. \ref{teleportationchain} shows the chain in the case $2\ell = 4$.  Note that $2\ell$ plays the role that the quantity $\ell$ played in the section 2.

\begin{remark}
To motivate the name ``teleportation'' chain Hamiltonian in this definition, we recall the quantum teleportation circuit depicted in Fig. \ref{teleportationcircuit}.  Its first gate directs the bottom two qubits into the incoming Bell pair state $(\ket{0} \ket{0} + \ket{1} \ket{1})/\sqrt{2}$.  This is followed by a Bell basis measurement.  If the result $(\ket{0} \ket{0} + \ket{1} \ket{1})/\sqrt{2}$ is obtained, the teleportation is successful.  Then, the bottom qubit exits the circuit with the incoming state $\ket{\mu}$.

The teleportation chain Hamiltonian is a parent Hamiltonian whose ground state emulates the quantum teleportation circuit.  Consider a simplified chain of 3 qutrits with Hamiltonian $ \textsf{I}  \otimes {H}^{B}+ {H}^{P}(\theta)\otimes \textsf{I} $ .  Two of its (unnormalized) ground states are ``history'' states 
\[
\cos \theta \ket{b} \frac{\ket{0} \ket{0} + \ket{1} \ket{1}}{\sqrt{2}} + \frac{\sin \theta}{2} \left|{\scriptscriptstyle IDLE}\right>\left|{\scriptscriptstyle IDLE}\right> \ket{b}
\]
for $b \in \mathbb{F}_2$; a history state contains terms corresponding to each time step in the history of the circuit.  The term ${H}^{B}$ drives adjacent spins $2$ and $3$ into the incoming Bell pair state $(\ket{0} \ket{0} + \ket{1} \ket{1})/\sqrt{2}$.   The term ${H}^{P}(\theta)$ augments the amplitude of the desired post-measurement state $(\ket{0} \ket{0} + \ket{1} \ket{1})/\sqrt{2}$ for spins $1$ and $2$.  It does this by driving the desired state into $\left|{\scriptscriptstyle IDLE}\right>\left|{\scriptscriptstyle IDLE}\right>$ with appropriate coefficients.  This has the effect of successfully teleporting the state of spin $1$ into spin $3$ with probability at least $(1/4) \sin^4 \theta$.

\begin{figure}[H]
\begin{center}
\includegraphics[width=4.5in]{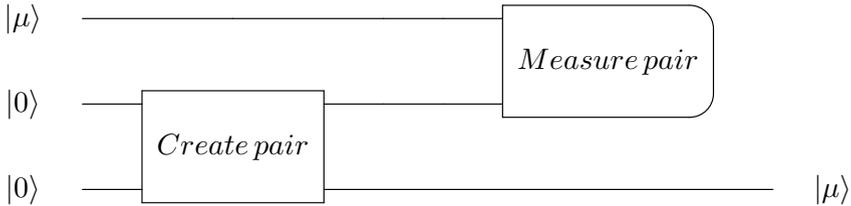}
\end{center}
\caption{Teleportation circuit diagram.  It is assumed that the measurement obtains the Bell state $(\ket{0} \ket{0} + \ket{1} \ket{1})/\sqrt{2}$ so that no Pauli correction is required for the outgoing state $\ket{\mu}$ to match the incoming state.}
\label{teleportationcircuit}
\end{figure}

\end{remark}

Although the terms of the Hamiltonian alternate between ${H}^{B}$ and $ {H}^{P}(\theta)$,  it is possible to put the Hamiltonian into the translationally invariant form (\ref{H}).  
\begin{definition}
Let a composite spin with dimension $d = 9$ be defined as as pair of adjacent qutrits.  Thus, a teleportation chain with $2\ell$ qutrits has $\ell$ composite spins and a boundary qubit at each end.  Set the coupling in (\ref{H}) to
\[
\bar{H} = \textsf{I}  \otimes  {H}^{B} \otimes \textsf{I}+ \textsf{I} \otimes \textsf{I} \otimes \frac{1}{2}{H}^{P}(\theta)+\frac{1}{2}{H}^{P}(\theta) \otimes \textsf{I} \otimes \textsf{I}.
\]
The term $\textsf{I} \otimes \textsf{I}\otimes \frac{1}{2}{H}^{P}$, for example, applies an identity to one composite 9-dimensional object and  $ \frac{1}{2}{H}^{P}$ to its neighboring composite object.    The boundary spins are qubits with dimension $d_0 = d_{\ell+1}  = 2$. Set the boundary terms of the Hamiltonian (\ref{mathcalH}) to
\[
\bar{H}^{0,1} =  \textsf{I}  \otimes {H}^{B}   +  \frac{1}{2}{H}^{P}(\theta) \otimes \textsf{I} 
\]
and
\[
\bar{H}^{\ell,\ell+1} =  {H}^{B} \otimes \textsf{I}  + \textsf{I} \otimes \frac{1}{2}{H}^{P}(\theta) .
\]
With these definitions in place, $\mathcal{H}$ of (\ref{mathcalH}) becomes the teleportation chain Hamiltonian, and we can apply our formalism to a chain of $\ell$ composite spins with a qubit at each end.  
\end{definition}

Direct application of the local gap method \cite{Knabe1988} or the Martingale method \cite{Nachtergaele1996} shows that the Hamiltonian is gapped for a range of $\theta$ but does not cover the entire range $0 \le \theta < \pi/2$.  Using renormalization will establish a gap throughout the entire range.

\subsection{Ground states of the teleportation chain}
\label{3SS2}

Divide the chain, ignoring at first the boundary qubits and boundary Hamiltonian terms, into segments of length ${\bar{\ell}}\ge 2$ composite spins of dimension $9$.  Our first order of business is to compute the ground states of each segment.  

\begin{definition}
Define the map
\begin{align}
M  = & \cos \theta (\ket{0}\bra{0}+\ket{1}\bra{1})\ \otimes  (\ket{0}\bra{0}+\ket{1}\bra{1}) \nonumber \\
 & +\frac{1}{\sqrt{2}} \sin \theta  (\ket{\scriptscriptstyle IDLE} \bra{0} \otimes  \ket{\scriptscriptstyle IDLE} \bra{0} + \ket{\scriptscriptstyle IDLE}  \bra{1}\otimes \ket{\scriptscriptstyle IDLE}  \bra{1})
\end{align}
that acts on 2 adjacent qutrits.  (Note that the second term carries $\left( \frac{\ket{0}\ket{0} +  \ket{1} \ket{1}}{\sqrt{2}} \right)$ to $\ket{\scriptscriptstyle IDLE}\ket{\scriptscriptstyle IDLE}$.)   Set 4 (non-orthogonal) states to
\begin{equation}
\ket{\psi_{a,b}(\bar{\ell})} =  M^{ \otimes {\bar{\ell}}}  \left[ \ket{a} \otimes \left( \frac{\ket{0}\ket{0} +  \ket{1} \ket{1}}{\sqrt{2}} \right)^{\otimes {\bar{\ell}}-1}   \otimes \ket{b}\right] 
 \label{psiab}
\end{equation}
where $a,b \in \mathbb{F}_2$.
\end{definition}

The following technical proposition will prove useful in dealing with these states.  We introduce the notation $Tr_i$ to denote the trace over the state of qutrit $i$ (not the trace over the composite spin $i$, which would be denoted by the trace over two qutrits $Tr_{2i-1,2i}$).

\begin{proposition}
\begin{align*}
& Tr_{2,\dots,2{\bar{\ell}}}  \ket{\psi_{a,b}(\bar{\ell})} \bra{\psi_{a^\prime,b^\prime}(\bar{\ell}) }  = \\
& \hspace{0.25in} \cos^2 \theta  (\alpha + \beta) \delta_{b,b^\prime} \ket{a}\bra{a^\prime} + \sin^2 \theta (\delta_{a,a^\prime} \delta_{b,b^\prime} \alpha/4 + \delta_{a,b}\delta_{a^\prime,b^\prime} \beta/2 ) \ket{\scriptscriptstyle IDLE} \bra{\scriptscriptstyle IDLE}
\end{align*}
where 
\begin{align*}
 \left[ \begin{array}{c} \alpha \\ \beta \end{array}\right]  & = \left[\begin{array}{cc}\cos ^2 \theta + (1/4) \sin ^2 \theta & \cos^2 \theta \\ 0 & (1/4) \sin ^2 \theta \end{array} \right]^{({\bar{\ell}}-1)}\left[\begin{array}{c} 0 \\ 1 \end{array} \right] \\
 & = \left[\begin{array}{c} (\cos ^2 \theta + (1/4) \sin ^2 \theta)^{({\bar{\ell}}-1)} - ((1/4) \sin ^2  \theta)^{({\bar{\ell}}-1)}  \\ ((1/4) \sin ^2  \theta)^{({\bar{\ell}}-1)} \end{array} \right]
\end{align*}
\end{proposition}

\begin{proof}
The proof is inductive.  We conjecture the form 
\[
Tr_{3,\dots,2{\bar{\ell}}} \ket{\psi_{a,b}(\bar{\ell})} \bra{\psi_{a^\prime,b^\prime}(\bar{\ell}) } = M \ket{a} \bra{a^\prime} \otimes \left[ (\alpha/2) \delta_{b,b^\prime}  \textsf{I}+ \beta \ket{b} \bra{b^\prime}\right] M^\dagger.
\]
for the trace over all the qutrits except the first 2 in the segment.  For the case ${\bar{\ell}} = 1$, where we take no trace, this form is obviously correct with $\alpha = 0$ and $\beta = 1$.  Assuming the form is correct for ${\bar{\ell}}$, we obtain the result for ${\bar{\ell}}+1$ by computing
\begin{align*}
Tr_{3,4} M \otimes M \bigg( \ket{a} \bra{a^\prime} \otimes & \frac{\ket{0} \ket{0} +  \ket{1} \ket{1}}{\sqrt{2}} \frac{\bra{0}  \bra{0} +  \bra{1}  \bra{1}}{\sqrt{2}}  \\
\otimes & \left[ (\alpha/2) \delta_{b,b^\prime}  \textsf{I}+ \beta \ket{b} \bra{b^\prime}\right] \bigg) M^\dagger \otimes  M^\dagger.
\end{align*}
We find that the trace returns us to our conjectured form with $\alpha \rightarrow (\cos ^2 \theta + (1/4) \sin ^2 \theta) \alpha + \cos^2 \theta \beta$ and $\beta \rightarrow (1/4) \sin ^2 \theta \beta$ as in the matrix equation stated in the proposition.  Taking $Tr_2$ proves the proposition.
\end{proof}

\begin{remark}
Note that only the final term, proportional to $\beta \delta_{a,b}\delta_{a^\prime,b^\prime}$, describes any correlations between ends of the segment.  This term decreases exponentially with ${\bar{\ell}}$ since $\beta$ decreases exponentially with ${\bar{\ell}}$.
\end{remark}

From this proposition, we can take $Tr_1$ and conclude that our 4 states $\ket{\psi_{a,b}}$ are not orthogonal.   A set of orthonormal ground states is given by the following lemma.

\begin{lemma}

For $0 \le \theta < \pi/2$, the 4 orthonormal states 
\[
 |\Psi_0(\bar{\ell})\rangle, |\Psi_1(\bar{\ell})\rangle, |\Psi_2(\bar{\ell})\rangle,|\Psi_3(\bar{\ell})\rangle\}
\]
span the space of zero energy eigenstates of a segment of length ${\bar{\ell}} \ge 2$ composite spins, where
\begin{align} \nonumber
|\Psi_0(\bar{\ell})\rangle &= \mathcal{N}_0(\bar{\ell}) (\ket{\psi_{0,0}(\bar{\ell})}+\ket{\psi_{1,1}(\bar{\ell})})/\sqrt{2} \\ \label{Psi} 
|\Psi_1(\bar{\ell})\rangle &= \mathcal{N}_1(\bar{\ell}) (\ket{\psi_{0,1}(\bar{\ell})}+\ket{\psi_{1,0}(\bar{\ell})})/\sqrt{2}\\ \nonumber
|\Psi_2(\bar{\ell})\rangle &= \mathcal{N}_2(\bar{\ell}) (\ket{\psi_{0,1}(\bar{\ell})}-\ket{\psi_{1,0}(\bar{\ell})})/\sqrt{2}\\ \nonumber
|\Psi_3(\bar{\ell})\rangle &= \mathcal{N}_3(\bar{\ell}) (\ket{\psi_{0,0}(\bar{\ell})}-\ket{\psi_{1,1}(\bar{\ell})})/\sqrt{2},
\end{align}
with normalization constants
\[
\mathcal{N}_0(\bar{\ell}) = \frac{1}{\sqrt{(\cos ^2 \theta + (1/4) \sin ^2 \theta)^{{\bar{\ell}}} +3 ((1/4) \sin ^2  \theta)^{{\bar{\ell}}}} }
\]
and
\[
\mathcal{N}_1(\bar{\ell}) =\mathcal{N}_2(\bar{\ell}) =\mathcal{N}_3(\bar{\ell}) =\mathcal{N} (\bar{\ell})= \frac{1}{\sqrt{(\cos ^2 \theta + (1/4) \sin ^2 \theta)^{{\bar{\ell}}} - ((1/4) \sin ^2  \theta)^{{\bar{\ell}}}}}.
\]
Thus, $\bar{P} = |\Psi_0(\bar{\ell})\rangle\bra{\Psi_0(\bar{\ell})}+|\Psi_1(\bar{\ell})\rangle\bra{\Psi_1(\bar{\ell})}+|\Psi_2(\bar{\ell})\rangle\bra{\Psi_2(\bar{\ell})}+|\Psi_3(\bar{\ell})\rangle\bra{\Psi_3(\bar{\ell})}$ and $\bar{d}=4$.
\label{4states}
\end{lemma}

\begin{remark}
Compared to $\mathcal{N}(\bar{\ell})$, the normalization constant $\mathcal{N}_0(\bar{\ell})$ includes an additional contribution $4 ((1/4) \sin ^2  \theta)^{{\bar{\ell}}}$ because $(\ket{\psi_{0,0}(\bar{\ell})}+\ket{\psi_{1,1}(\bar{\ell})})/\sqrt{2} $ includes an extra term $2 ((1/2) \sin \theta)^{{\bar{\ell}}} \ket{\scriptscriptstyle IDLE} \otimes \dots \otimes \ket{\scriptscriptstyle IDLE}$.  
\end{remark}

\begin{proof}
To confirm that the states (\ref{psiab}) are zero energy eigenstates of the Hamiltonian, note that $ {H}^{P}(\theta) M = 0$.  Similarly, note that 
the Bell pairs in parentheses in (\ref{psiab}) ensure that the state is annihilated by all the ${H}^{B}$ terms in $H^S$.  Thus, the states (\ref{psiab}) are annihilated by the Hamiltonian, so the linear combinations (\ref{Psi}) are as well.  We confirm that the $\ket{\Psi_z}$ are orthonormal as follows.  Taking $Tr_1$ of the formula in the previous proposition yields the following matrix of overlaps $\left\langle \psi_{a^\prime,b^\prime} (\bar{\ell})| \psi_{a,b}(\bar{\ell}) \right\rangle$
\[
\left[\begin{array}{cccc}
\alpha + 2\beta & 0 & 0 & 2 \beta \\
0 & \alpha  & 0 & 0 \\
0 & 0 & \alpha & 0 \\
 2 \beta  & 0 & 0 & \alpha + 2\beta
\end{array} \right]
\]
where $\alpha = (\cos ^2 \theta + (1/4) \sin ^2 \theta)^{{\bar{\ell}}} - ((1/4) \sin ^2  \theta)^{{\bar{\ell}}}$ and $\beta = ((1/4) \sin ^2  \theta)^{{\bar{\ell}}}$.
The overlap matrix has eigenvectors 
\[
\frac{\mathcal{N}_0(\bar{\ell})}{\sqrt{2}}\left[\begin{array}{c}1 \\ 0 \\0 \\ 1\end{array}\right],\, \frac{\mathcal{N}(\bar{\ell})}{\sqrt{2}}\left[\begin{array}{c}0 \\ 1 \\1 \\ 0\end{array}\right],\, \frac{\mathcal{N}(\bar{\ell})}{\sqrt{2}}\left[\begin{array}{c}0 \\ 1 \\-1 \\ 0\end{array}\right],\, \frac{\mathcal{N}(\bar{\ell})}{\sqrt{2}}\left[\begin{array}{c}1 \\ 0 \\0 \\ -1\end{array}\right]
\]
and eigenvalues  $1/\mathcal{N}_0^2(\bar{\ell})$, $1/\mathcal{N}^2(\bar{\ell})$, $1/\mathcal{N}^2(\bar{\ell})$, and $1/\mathcal{N}^2(\bar{\ell})$ respectively.   The states (\ref{Psi}) are defined according to these eigenvectors.  The matrix element of the overlap matrix between two distinct eigenvectors vanishes,  so the $\ket{\Psi_z(\bar{\ell})}$ are orthogonal.  The expectation value of the overlap matrix is 1 for every eigenvector, so the $\ket{\Psi_z(\bar{\ell})}$ are normalized.

To show that these states span the space of zero energy eigenstates, we make an inductive argument.  Start with a segment of length ${\bar{\ell}} = 2$ composite spins, which is $4$ qutrits.   We want to show that $\bar{H}$ has 4 zero energy eigenstates.  Consider the left 2 qutrits.  Their state must be annihilated by $(1/2){H}^{P}(\theta)$.  Of their $3^2=9$ basis states, $2+2=4$ are raised in energy by the first 2 sums in ${H}^{P}(\theta)$.  This leaves $9-4=5$ basis states: the 4 Bell basis states and $\ket{\scriptscriptstyle IDLE} \ket{\scriptscriptstyle IDLE}$.  The final projection in ${H}^{P}(\theta)$ forces $\frac{\ket{0} \ket{0}+\ket{1} \ket{1}}{\sqrt{2}}$ and $\ket{\scriptscriptstyle IDLE} \ket{\scriptscriptstyle IDLE}$ to occur in the combination $\cos \theta  \frac{\ket{0} \ket{0}+\ket{1} \ket{1}}{\sqrt{2}} +\sin \theta \ket{{\scriptscriptstyle IDLE}} \ket{{\scriptscriptstyle IDLE}}$, leaving 4 allowed states
\begin{align} \nonumber
\left\{\cos \theta  \frac{\ket{0} \ket{0}+\ket{1} \ket{1}}{\sqrt{2}} +\sin \theta \ket{{\scriptscriptstyle IDLE}} \ket{{\scriptscriptstyle IDLE}},  \right. &\frac{\ket{0} \ket{1}+\ket{1} \ket{0}}{\sqrt{2}},  \\ 
  \frac{\ket{0} \ket{1}-\ket{1} \ket{0}}{\sqrt{2}}, & \left. \frac{\ket{0} \ket{0}-\ket{1} \ket{1}}{\sqrt{2}}\right\}. \label{basis4}
\end{align}
By the same argument, the right $4$ qutrits have the same 4 allowed states.
Thus, we need to consider a total of $4^2 = 16$ basis states.  We compute the form of $\textsf{I} \otimes {H}^{B} \otimes \textsf{I}$ in this basis.  At $\theta = 0$, our basis reduces to 
\begin{equation}
\left\{  \frac{\ket{0}   \ket{0} +  \ket{1}   \ket{1}}{\sqrt{2}} , \frac{\ket{0} \ket{1}+\ket{1} \ket{0}}{\sqrt{2}},  \frac{\ket{0} \ket{1}-\ket{1} \ket{0}}{\sqrt{2}}, \frac{\ket{0}   \ket{0} -  \ket{1}   \ket{1}}{\sqrt{2}}\right\}^{\otimes 2}
\label{Bellsquared}
\end{equation}
and $\textsf{I} \otimes {H}^{B} \otimes \textsf{I}$ becomes
{\setlength{\extrarowheight}{2pt}
\begin{align}
K = \left[
\begin{array}{*{16}{R{1.5em}}}
 \frac{3}{4} & 0 & 0 & 0 & 0 & -\frac{1}{4} & 0 & 0 & 0 & 0 & \frac{1}{4} & 0 & 0 & 0 & 0 & -\frac{1}{4} \\
 0 & \frac{3}{4} & 0 & 0 & -\frac{1}{4} & 0 & 0 & 0 & 0 & 0 & 0 & \frac{1}{4} & 0 & 0 & -\frac{1}{4} & 0 \\
 0 & 0 & \frac{3}{4} & 0 & 0 & 0 & 0 & \frac{1}{4} & -\frac{1}{4} & 0 & 0 & 0 & 0 & -\frac{1}{4} & 0 & 0 \\
 0 & 0 & 0 & \frac{3}{4} & 0 & 0 & \frac{1}{4} & 0 & 0 & -\frac{1}{4} & 0 & 0 & -\frac{1}{4} & 0 & 0 & 0 \\
 0 & -\frac{1}{4} & 0 & 0 & \frac{3}{4} & 0 & 0 & 0 & 0 & 0 & 0 & \frac{1}{4} & 0 & 0 & -\frac{1}{4} & 0 \\
 -\frac{1}{4} & 0 & 0 & 0 & 0 & \frac{3}{4} & 0 & 0 & 0 & 0 & \frac{1}{4} & 0 & 0 & 0 & 0 & -\frac{1}{4} \\
 0 & 0 & 0 & \frac{1}{4} & 0 & 0 & \frac{3}{4} & 0 & 0 & \frac{1}{4} & 0 & 0 & \frac{1}{4} & 0 & 0 & 0 \\
 0 & 0 & \frac{1}{4} & 0 & 0 & 0 & 0 & \frac{3}{4} & \frac{1}{4} & 0 & 0 & 0 & 0 & \frac{1}{4} & 0 & 0\\
  0 & 0 &-\frac{1}{4}& 0 & 0 & 0 & 0 & \frac{1}{4} & \frac{3}{4} & 0 & 0 & 0 & 0 & -\frac{1}{4} & 0 & 0 \\
 0 & 0 & 0 & -\frac{1}{4} & 0 & 0 & \frac{1}{4} & 0 & 0 & \frac{3}{4} & 0 & 0 & -\frac{1}{4} & 0 & 0 & 0 \\
 \frac{1}{4} & 0 & 0 & 0 & 0 & \frac{1}{4} & 0 & 0 & 0 & 0 & \frac{3}{4} & 0 & 0 & 0 & 0 & \frac{1}{4} \\
 0 & \frac{1}{4} & 0 & 0 & \frac{1}{4} & 0 & 0 & 0 & 0 & 0 & 0 & \frac{3}{4} & 0 & 0 & \frac{1}{4} & 0 \\
 0 & 0 & 0 & -\frac{1}{4} & 0 & 0 & \frac{1}{4} & 0 & 0 & -\frac{1}{4} & 0 & 0 & \frac{3}{4} & 0 & 0 & 0 \\
 0 & 0 & -\frac{1}{4} & 0 & 0 & 0 & 0 & \frac{1}{4} & -\frac{1}{4} & 0 & 0 & 0 & 0 & \frac{3}{4} & 0 & 0 \\
 0 & -\frac{1}{4} & 0 & 0 & -\frac{1}{4} & 0 & 0 & 0 & 0 & 0 & 0 & \frac{1}{4} & 0 & 0 & \frac{3}{4} & 0 \\
 -\frac{1}{4} & 0 & 0 & 0 & 0 & -\frac{1}{4} & 0 & 0 & 0 & 0 & \frac{1}{4} & 0 & 0 & 0 & 0 & \frac{3}{4} 
\end{array}
\right]
\label{K}
\end{align}
}

Away from $\theta = 0$, we require some adjustments because of normalization: there is an extra factor of $\cos \theta$ in (\ref{basis4}) compared to (\ref{Bellsquared}).  We obtain
\begin{align}
\left[\begin{array}{cccc}
\cos \theta & 0 & 0 & 0 \\ 
0 & 1 & 0 & 0 \\ 
0 & 0 & 1 & 0  \\ 
0 & 0 & 0 & 1 \end{array} \right] ^{\otimes 2} 
K
\left[\begin{array}{cccc}
\cos \theta & 0 & 0 & 0 \\ 
0 & 1 & 0 & 0 \\ 
0 & 0 & 1 & 0  \\ 
0 & 0 & 0 &1 \end{array} \right] ^{\otimes 2}
\label{cosK}
\end{align}
by computing the matrix elements 
\begin{align*}
&\bra{\phi_i} \bra{\phi_j} \textsf{I}  \otimes {H}^{B}  \otimes \textsf{I}  \ket{\phi_{i^\prime}} \ket{\phi_{j^\prime}} \\
& = Tr_{2,3} {H}^{B}  \Big(Tr_1 \ket{\phi_{i^\prime}} \bra{\phi_i} \Big)\Big(Tr_4 \ket{\phi_{j^\prime}} \bra{\phi_j} \Big).
\end{align*}
where $\ket{\phi_{i}}$ denotes an element of  (\ref{basis4}).  

As long as $\cos \theta \ne 0$, the dimension of the kernel of this matrix equals that of $K$.  Diagonalizing $K$, one finds that its kernel has dimension 4.  Thus, a segment of length ${\bar{\ell}}=2$ has a 4 dimensional space of zero energy eigenstates, which must be spanned by our 4 states.  (Roughly speaking, each of the 3 terms in ${H}^{B}$ pushes up the energy of 4 states when we consider $\textsf{I} \otimes {H}^{B} \otimes \textsf{I}$.  This leaves us with $16-3 \times 4 = 4$ zero energy eigenstates (\ref{Psi}).)

Assuming the inductive hypothesis for ${\bar{\ell}}$ composite spins, we now show it for ${\bar{\ell}}+1$ composite spins.  When we add $2$ new qutrits (1 composite spin) to the segment, they are constrained to lie in a state $ \ket{\phi_i}$ belonging to the 4 dimensional basis (\ref{basis4}).  The segment of length ${\bar{\ell}}$ has 4 zero energy states $\{ |\Psi_0(\bar{\ell})\rangle,|\Psi_1(\bar{\ell})\rangle,|\Psi_2(\bar{\ell})\rangle,|\Psi_3(\bar{\ell})\rangle \}$.  Thus, we have a $4^2 = 16$ dimensional basis of states.  The Hamiltonian $\textsf{I} \otimes {H}^{B} \otimes \textsf{I}$ between the new qutrits and the length ${\bar{\ell}}$ segment is computed  using
\begin{align*}
&\bra{\phi_i} \langle\Psi_j (\bar{\ell})| \textsf{I}  \otimes {H}^{B}  \otimes \textsf{I}  \ket{\phi_{i^\prime}}  \ket{\Psi_{j^\prime}(\bar{\ell})} \nonumber \\
& = Tr_{2,3} {H}^{B}  \Big(Tr_1 \ket{\phi_{i^\prime}} \bra{\phi_i} \Big) \Big( Tr_{4,\dots,2 \ell + 2}   \ket{\Psi_{j^\prime}(\bar{\ell})}\langle\Psi_j(\bar{\ell}) | \Big)
\end{align*}
and inserting definition (\ref{Psi}) into proposition 3.  The result takes the form
\begin{align}
& \cos^2 \theta \left(\cos ^2\theta+\frac{\sin ^2\theta}{4}\right)^{{\bar{\ell}}-1} \label{inductiveK} \\
&\hspace{0.25in} \left( \left[\begin{array}{cccc}
\cos \theta & 0 & 0 & 0 \\ 
0 & 1 & 0 & 0 \\ 
0 & 0 & 1 & 0  \\ 
0 & 0 & 0 & 1 \end{array} \right]   \otimes
\left[\begin{array}{cccc}
\mathcal{N}_0(\bar{\ell}) & 0 & 0 & 0 \\ 
0 & \mathcal{N} (\bar{\ell})& 0 & 0 \\ 
0 & 0 & \mathcal{N} (\bar{\ell})& 0  \\ 
0 & 0 & 0 &\mathcal{N}(\bar{\ell}) \end{array} \right]  \right) K  \nonumber \\
&
\hspace{0.75in} \left(\left[\begin{array}{cccc}
\cos \theta & 0 & 0 & 0 \\ 
0 & 1 & 0 & 0 \\ 
0 & 0 & 1 & 0  \\ 
0 & 0 & 0 & 1 \end{array} \right]  \otimes
\left[\begin{array}{cccc}
\mathcal{N}_0(\bar{\ell}) & 0 & 0 & 0 \\ 
0 & \mathcal{N}(\bar{\ell}) & 0 & 0 \\ 
0 & 0 & \mathcal{N}(\bar{\ell}) & 0  \\ 
0 & 0 & 0 &\mathcal{N}(\bar{\ell}) \end{array} \right]\right).
\nonumber
\end{align}
Again, because the dimension of the kernel of $K$ is 4, the dimension of the kernel of (\ref{inductiveK}) is 4 provided $\cos \theta \ne 0$.
\end{proof}

Now that we have defined $\bar{P}$, we can compute the renormalized coupling $\bar{h}$ and confirm that the chain satisfies Def. \ref{decayingcorrelations}.

\subsection{The renormalized teleportation chain and decaying correlation}
\label{3SS3}

\begin{lemma}
For $0 \le \theta < \pi/2$, the ground states of the teleportation chain exhibit decaying correlations.
\label{decayingcorrelationsteleportationchain}
\end{lemma}
\begin{proof}
We evaluate $\bar{h}$ and apply Def. \ref{decayingcorrelations}.  The required calculation is familiar from the proof of the previous lemma.  Our basis of states is  $\left\{|\Psi_0 (\bar{\ell})\rangle, |\Psi_1(\bar{\ell})\rangle, |\Psi_2(\bar{\ell})\rangle,|\Psi_3(\bar{\ell})\rangle\right\} \otimes \left\{ |\Psi_0(\bar{\ell})\rangle, |\Psi_1(\bar{\ell})\rangle, |\Psi_2(\bar{\ell})\rangle,|\Psi_3(\bar{\ell})\rangle\right\}$.   At $\theta = 0$, the calculation reduces to computing $\textsf{I} \otimes {H}_{B} \otimes \textsf{I}$ in the Bell basis (\ref{Bellsquared}).
Thus, in this basis,
\[
\bar{h}\vert_{\theta = 0} = \tilde{h}  \vert_{\theta = 0} = K.
\]
(Here, we have abused notation by identifying operators $\bar{h}\vert_{\theta = 0} = \tilde{h}  \vert_{\theta = 0}$ with their matrix representation in the the $\bar{d}^2 = 16$ dimensional basis of ground states.  We continue to abuse notation in this way in the following where no confusion will arise.) 

Away from $\theta = 0$, we need some adjustments because of normalization factors.  We define $\tilde{h}$ to be
\begin{equation}
\tilde{h} = \mathcal{N}^4(\bar{\ell}) \cos^4 \theta \left(\cos ^2\theta+\frac{\sin ^2\theta}{4}\right)^{2{\bar{\ell}}-2} \tilde{h}\vert_{\theta = 0}
\label{tildeh}
\end{equation}

Is it convenient to set  $\mathcal{N}_0(\bar{\ell})\equiv \mathcal{N} (\bar{\ell}) (1-\Delta(\bar{\ell})/4)$ where
\begin{equation}
\frac{\Delta(\bar{\ell})}{4} = 1 - \sqrt{\frac{(\cos ^2 \theta + (1/4) \sin ^2 \theta)^{{\bar{\ell}}} - ((1/4) \sin ^2  \theta)^{{\bar{\ell}}}}{(\cos ^2 \theta + (1/4) \sin ^2 \theta)^{{\bar{\ell}}} +3 ((1/4) \sin ^2  \theta)^{{\bar{\ell}}}}}. \label{barUpsilon}
\end{equation}
Note that $\Delta(\bar{\ell})$ decreases exponentially to zero in ${\bar{\ell}}$ at any fixed $0 < \theta < \pi/2$ and is identically zero at $\theta = 0$.  
We have
\begin{align*}
\bar{h}  = \left[\begin{array}{cccc} 1-\frac{\Delta(\bar{\ell})}{4} & 0 & 0 & 0 \\ 0 & 1 & 0 & 0 \\ 0 & 0 & 1 & 0  \\ 0 & 0 & 0 &1 \end{array} \right] ^{\otimes 2}  \tilde{h}  \left[\begin{array}{cccc} 1-\frac{\Delta(\bar{\ell})}{4} & 0 & 0 & 0 \\ 0 & 1 & 0 & 0 \\ 0 & 0 & 1 & 0  \\ 0 & 0 & 0 &1 \end{array} \right] ^{\otimes 2} \equiv \tilde{h} + {\bar{k}}.
\end{align*}

We can evaluate the norm of ${\bar{k}}$  according to
\begin{align}
&\norm{{\bar{k}}}_2  = \norm{\bar{h}-\tilde{h}} _2= \nonumber\\
&\frac{1}{2} \left\lVert \left( \left[\begin{array}{cccc} 1-\frac{\Delta(\bar{\ell})}{4} & 0 & 0 & 0 \\ 0 & 1 & 0 & 0 \\ 0 & 0 & 1 & 0  \\ 0 & 0 & 0 &1 \end{array} \right] ^{\otimes 2}-\textsf{I}\otimes\textsf{I}\right) \tilde{h}  \left(  \left[\begin{array}{cccc} 1-\frac{\Delta(\bar{\ell})}{4} & 0 & 0 & 0 \\ 0 & 1 & 0 & 0 \\ 0 & 0 & 1 & 0  \\ 0 & 0 & 0 &1 \end{array} \right] ^{\otimes 2}+\textsf{I}\otimes\textsf{I}\right) \nonumber   \right. \\
& \hspace{0.5in} +\left.  \left( \left[\begin{array}{cccc} 1-\frac{\Delta(\bar{\ell})}{4} & 0 & 0 & 0 \\ 0 & 1 & 0 & 0 \\ 0 & 0 & 1 & 0  \\ 0 & 0 & 0 &1 \end{array} \right] ^{\otimes 2}+\textsf{I}\otimes\textsf{I}\right) \tilde{h}  \left(  \left[\begin{array}{cccc} 1-\frac{\Delta(\bar{\ell})}{4} & 0 & 0 & 0 \\ 0 & 1 & 0 & 0 \\ 0 & 0 & 1 & 0  \\ 0 & 0 & 0 &1 \end{array} \right] ^{\otimes 2}-\textsf{I}\otimes\textsf{I}\right) \right\rVert_2\nonumber  \\
& \le  \left\lVert\left( \left[\begin{array}{cccc} 1-\frac{\Delta(\bar{\ell})}{4} & 0 & 0 & 0 \\ 0 & 1 & 0 & 0 \\ 0 & 0 & 1 & 0  \\ 0 & 0 & 0 &1 \end{array} \right] ^{\otimes 2}-\textsf{I}\otimes\textsf{I}\right) \right\rVert_2  \norm{\tilde{h}}_2 \left\lVert\left( \left[\begin{array}{cccc} 1-\frac{\Delta(\bar{\ell})}{4} & 0 & 0 & 0 \\ 0 & 1 & 0 & 0 \\ 0 & 0 & 1 & 0  \\ 0 & 0 & 0 &1 \end{array} \right] ^{\otimes 2}+\textsf{I}\otimes\textsf{I}\right) \right\rVert_2  \nonumber \\
& =2 \left(1-\left(1-\frac{\Delta(\bar{\ell})}{4}\right)^2\right) \norm{\tilde{h}}_2 \le \Delta(\bar{\ell}) \norm{\tilde{h}}_2 \nonumber \\
&= 
\Delta(\bar{\ell})  \mathcal{N}^4(\bar{\ell}) \cos^4 \theta \left(\cos ^2\theta+\frac{\sin ^2\theta}{4}\right)^{2{\bar{\ell}}-2} \le \Delta(\bar{\ell}).\label{ksize} 
\end{align}

We see that, as required by definition \ref{decayingcorrelations}, it is possible to ensure $\norm{{\bar{k}}}_2 < \Delta$ for any $\Delta > 0$ by choosing $\bar{\ell}$ large enough to make $\Delta (\bar{\ell})$ sufficiently small.

Finally, we can confirm that the commutator of $\tilde{h} \otimes \bar{P}$ and $\bar{P} \otimes \tilde{h}$ vanishes.  In the $4^3 = 64$ dimensional basis of zero energy eigenstates of three adjacent segments, we have 
\[
\bar{P} \otimes \tilde{h}  = \mathcal{N}^4 (\bar{\ell})\cos^4 \theta \left(\cos ^2\theta+\frac{\sin ^2\theta}{4}\right)^{2{\bar{\ell}}-2}  \left[\begin{array}{cccc} 1 & 0 & 0 & 0 \\ 0 & 1 & 0 & 0 \\ 0 & 0 & 1 & 0  \\ 0 & 0 & 0 &1 \end{array} \right]  \otimes K.
\]
and
\[
\tilde{h} \otimes \bar{P} = \mathcal{N}^4(\bar{\ell}) \cos^4 \theta \left(\cos ^2\theta+\frac{\sin ^2\theta}{4}\right)^{2{\bar{\ell}}-2}  K \otimes  \left[\begin{array}{cccc} 1 & 0 & 0 & 0 \\ 0 & 1 & 0 & 0 \\ 0 & 0 & 1 & 0  \\ 0 & 0 & 0 &1 \end{array} \right].
\]
We transform each segment from the Bell basis to the standard basis
\[
\left\{\ket{0} \ket{0},\ket{0} \ket{1},\ket{1} \ket{0},\ket{1} \ket{1}\right\}
\]
via a unitary operator $U$ .  Since 
\[
(U \otimes U) K (U^\dagger \otimes U^\dagger) = \left[\begin{array}{cc}1 & 0 \\ 0 & 1 \end{array} \right] \otimes \left[\begin{array}{rrrr} \frac{1}{2} &\,\, 0 &\,\, 0 & -\frac{1}{2} \\ 0 & 1 & 0 & 0 \\ 0 & 0 & 1 & 0  \\ -\frac{1}{2} & 0 & 0 &\frac{1}{2} \end{array} \right] \otimes \left[\begin{array}{cc}1 & 0 \\ 0 & 1 \end{array} \right],
\]
we have
\begin{align*}
U \otimes U \otimes U & \,\,  [\tilde{h} \otimes \bar{P} ,\bar{P} \otimes \tilde{h}] \,\, U^\dagger \otimes U^\dagger \otimes U^\dagger \\
& =  \mathcal{N}^8(\bar{\ell}) \cos^8 \theta \left(\cos ^2\theta+\frac{\sin ^2\theta}{4}\right)^{4{\bar{\ell}}-4}   \times \\
& \left[  \left[\begin{array}{cccc} 1 & 0 & 0 & 0 \\ 0 & 1 & 0 & 0 \\ 0 & 0 & 1 & 0  \\ 0 & 0 & 0 &1 \end{array} \right]  \otimes \left[\begin{array}{cc}1 & 0 \\ 0 & 1 \end{array} \right]\otimes \left[\begin{array}{rrrr} \frac{1}{2} &\,\, 0 &\,\, 0 & -\frac{1}{2} \\ 0 & 1 & 0 & 0 \\ 0 & 0 & 1 & 0  \\ -\frac{1}{2} & 0 & 0 &\frac{1}{2} \end{array} \right] \otimes \left[\begin{array}{cc}1 & 0 \\ 0 & 1 \end{array} \right],\right.\\
&\hspace{0.75in} \left. \left[\begin{array}{cc}1 & 0 \\ 0 & 1 \end{array} \right] \otimes \left[\begin{array}{rrrr} \frac{1}{2} &\,\, 0 &\,\, 0 & -\frac{1}{2} \\ 0 & 1 & 0 & 0 \\ 0 & 0 & 1 & 0  \\ -\frac{1}{2} & 0 & 0 &\frac{1}{2} \end{array} \right] \otimes \left[\begin{array}{cc}1 & 0 \\ 0 & 1 \end{array} \right] \otimes \left[\begin{array}{cccc} 1 & 0 & 0 & 0 \\ 0 & 1 & 0 & 0 \\ 0 & 0 & 1 & 0  \\ 0 & 0 & 0 &1 \end{array} \right]\right]  = 0.
\end{align*}
\end{proof}

\begin{lemma}
The teleportation chain Hamiltonian $H$ with open boundary conditions is gapped for $0 \le \theta < \pi/2$.
\end{lemma}
\begin{proof}
We apply theorem 2.  The previous lemma establishes decaying correlations, so now we need to bound the gap $\bar{g}$ of $\bar{h}$ from below.  The matrix (\ref{tildeh}) for $\tilde{h}$ can be diagonalized directly to identify 4 vanishing eigenvalues and 12 eigenvalues that all equal its gap
$\mathcal{N}^4(\bar{\ell}) \cos^4 \theta \left(\cos ^2\theta+\frac{\sin ^2\theta}{4}\right)^{2{\bar{\ell}}-2}  \ge \cos^4 \theta $.  
Note that $\bar{h}$ equals $\tilde{h}$ flanked by invertible diagonal matrices, so $\tilde{h}$ has the same number of vanishing eigenvalues as $\bar{h}$.  We can therefore use the eigenvalue stability inequality to conclude $\bar{g} \ge \cos^4 \theta - \norm{\bar{k}}_2$. This is greater than, say, $(1/2) \cos^4 \theta$ for sufficiently large ${\bar{\ell}}$ that ensures $\norm{\bar{k}}_2$ is small.

Next, we see that the dimension of the kernel of $\bar{P} \otimes \bar{h} + \bar{h} \otimes \bar{P}$ equals the dimension of the kernel of $\bar{P} \otimes \tilde{h} + \tilde{h} \otimes \bar{P}$.  After all, a segment of length $3 {\bar{\ell}}$ has 4 zero-energy eigenstates by lemma 3, so the dimension of the kernel of $\bar{P} \otimes \bar{h} + \bar{h} \otimes \bar{P}$ is $4$.    This is true for all values of $\theta$, including $\theta = 0$.  Since $\bar{P} \otimes \bar{h} + \bar{h} \otimes \bar{P}$ reduces to $\bar{P} \otimes \tilde{h} + \tilde{h} \otimes \bar{P}$ at $\theta = 0$, it follows that the dimension of the kernel of $\bar{P} \otimes \tilde{h} + \tilde{h} \otimes \bar{P}$ is also 4.  Therefore, $\lambda_4 (\bar{P} \otimes \tilde{h} + \tilde{h} \otimes \bar{P})$ is at least the gap of $\tilde{h}$, which is greater than $\cos^4 \theta$.  The lemma then follows from theorem 2.
\end{proof}

\subsection{Concluding that the teleportation chain is gapped}
\label{3SS4}

We have proven so far that a teleportation chain Hamiltonian of the form (\ref{H}) with open ends has 4 zero-energy ground states and is gapped.  To finish our proof, we now establish that (\ref{mathcalH}) is gapped.

\begin{theorem}
The teleportation chain Hamiltonian ${\mathcal H}$ is gapped for $0 \le \theta < \pi/2$.
\end{theorem}
\begin{proof}
To show that ${\mathcal H}$ is gapped, we use expression (\ref{othermathcalH}) and employ the strategy of remark \ref{mathcalHremark}.  We set $B = H \otimes \textsf{I}^{\otimes \bar{\bar{\ell}}} +\textsf{I}^{\otimes \lfloor \ell/{\bar{\ell}} \rfloor \bar{\ell}} \otimes  H^R$ and $C =   \textsf{I}^{\otimes (\lfloor \ell/{\bar{\ell}} \rfloor \bar{\ell}-1) } \otimes \bar{H} \otimes \textsf{I}^{\otimes \bar{\bar{\ell}}-1}$ and apply lemma \ref{ABC}.  The previous lemma has shown that $H$ is gapped.  Suppose the lower bound of its gap is $g > 0$.  If the gap of $H^R$ as a function of its length $\bar{\bar{\ell}}$ is $g^R(\bar{\bar{\ell}}) > 0$, then the gap of $B$ is bounded below by $\min(g, \min_{0 \le \bar{\bar{\ell}} < \bar{\ell}} g^R(\bar{\bar{\ell}})) > 0$.  Thus, $B$ is gapped.  Note that $\norm{C}_2 \le 2$.  To apply lemma \ref{ABC}, our remaining task is to evaluate $C$ in the basis of ground states of $B$ and show that the result is gapped; then lemma \ref{ABC} implies that $A = B+C$ is gapped.

Using lemma \ref{4states}, we see that the basis of ground states of $B$ is 
\[
\left\{ |\Psi_0(\breve{\ell})\rangle ,|\Psi_1(\breve{\ell})\rangle ,|\Psi_2(\breve{\ell})\rangle ,|\Psi_3(\breve{\ell})\rangle \right\} \otimes \left\{ |\Psi_0(\bar{\bar{\ell}})\rangle,|\Psi_1(\bar{\bar{\ell}})\rangle,|\Psi_2(\bar{\bar{\ell}})\rangle,|\Psi_3(\bar{\bar{\ell}})\rangle \right\}.
\]
where $\breve{\ell} = \lfloor \ell/{\bar{\ell}} \rfloor \bar{\ell}$.  Evaluating $C$ in this basis proceeds like the proof of lemma \ref{decayingcorrelationsteleportationchain}.  The matrix elements of $C$ are
\begin{align*}
&\langle \Psi_i (\breve{\ell}) | \langle \Psi_j(\bar{\bar{\ell}})| \textsf{I}^{\otimes \breve{\ell}-1 } \otimes \bar{H} \otimes \textsf{I}^{\otimes \bar{\bar{\ell}}-1}   |\Psi_{i^\prime} (\breve{\ell})\rangle |\Psi_{j^\prime}(\bar{\bar{\ell}})\rangle \nonumber \\
& = Tr_{2 \bar{\bar{\ell}}-1,...,2\bar{\bar{\ell}}+2} \bar{H}   \Big( Tr_{2\bar{\bar{\ell}}+3,\dots,2\bar{\bar{\ell}} + 2 \breve{\ell} }   |\Psi_{i^\prime} (\breve{\ell}) \rangle \langle \Psi_i (\breve{\ell}) |  \Big)   \Big(Tr_{1,..., 2 \bar{\bar{\ell}}-2}  |\Psi_{j^\prime}(\bar{\ell})\rangle \langle\Psi_j(\bar{\ell}) | \Big) \\
& = Tr_{2 \bar{\bar{\ell}},2\bar{\bar{\ell}}+1} {H}^{B}   \Big( Tr_{2\bar{\bar{\ell}}+2,\dots,2\bar{\bar{\ell}} + 2 \breve{\ell} }   |\Psi_{i^\prime} (\breve{\ell}) \rangle \langle \Psi_i (\breve{\ell}) |  \Big)   \Big(Tr_{1,..., 2 \bar{\bar{\ell}}-1}  |\Psi_{j^\prime}(\bar{\ell})\rangle \langle\Psi_j(\bar{\ell}) | \Big)
\end{align*}
leading to a matrix of the form
\begin{align*}
&\mathcal{N}^2(\breve{\ell}) \mathcal{N}^2(\bar{\ell}) \cos^4 \theta \left(\cos ^2\theta+\frac{\sin ^2\theta}{4}\right)^{{\breve{\ell}}+\bar{\bar{\ell}}-2}  \\
&\hspace{0.25in} \left(\left[\begin{array}{cccc} 1-\frac{\Delta(\breve{\ell})}{4} & 0 & 0 & 0 \\ 0 & 1 & 0 & 0 \\ 0 & 0 & 1 & 0  \\ 0 & 0 & 0 &1 \end{array} \right] \otimes  \left[\begin{array}{cccc} 1-\frac{\Delta(\bar{\bar{\ell}})}{4} & 0 & 0 & 0 \\ 0 & 1 & 0 & 0 \\ 0 & 0 & 1 & 0  \\ 0 & 0 & 0 &1 \end{array} \right]\right) K  \\
&\hspace{0.5in} \left(\left[\begin{array}{cccc} 1-\frac{\Delta(\breve{\ell})}{4} & 0 & 0 & 0 \\ 0 & 1 & 0 & 0 \\ 0 & 0 & 1 & 0  \\ 0 & 0 & 0 &1 \end{array} \right] \otimes\left[\begin{array}{cccc} 1-\frac{\Delta(\bar{\bar{\ell}})}{4} & 0 & 0 & 0 \\ 0 & 1 & 0 & 0 \\ 0 & 0 & 1 & 0  \\ 0 & 0 & 0 &1 \end{array} \right] \right).
\end{align*}
Here, $K$ is defined in (\ref{K}) while $\Delta(\breve{\ell})$ and $\Delta(\bar{\bar{\ell}})$ are given by equation (\ref{barUpsilon}) with $\breve{\ell}$ and $\bar{\bar{\ell}}$ substituted for $\bar{\ell}$ respectively.
Our matrix has 4 zero energy ground states, just like $K$.  If not for the matrices flanking $K$, it would be easy to compute the first excited energy 
\[
\mathcal{N}^2(\breve{\ell}) \mathcal{N}^2(\bar{\ell}) \cos^4 \theta \left(\cos ^2\theta+\frac{\sin ^2\theta}{4}\right)^{{\breve{\ell}}+\bar{\bar{\ell}}-2} \ge \cos^4 \theta.
\]
Fortunately, arguing as in (\ref{ksize}), we see these flanking matrices only change the matrix by a small matrix of norm less than $(\Delta(\breve{\ell}) + \Delta(\bar{\bar{\ell}}))/2$.  So, the eigenvalue stability theorem ensures that the gap of our matrix can be made no smaller than, say, $(1/2) \cos^4 \theta$.

So, we have concluded that $H \otimes \textsf{I}^{\otimes \bar{\bar{\ell}}} +\textsf{I}^{\lfloor \ell/{\bar{\ell}} \rfloor \bar{\ell}-1 } \otimes \bar{H} \otimes \textsf{I}^{\otimes \bar{\bar{\ell}}-1}+\textsf{I}^{\otimes \lfloor \ell/{\bar{\ell}} \rfloor \bar{\ell}} \otimes  H^R$ is gapped.  This Hamiltonian has 4 ground states  $\left\{ |\Psi_0 (\ell)\rangle ,|\Psi_1(\ell)\rangle ,|\Psi_2(\ell)\rangle ,|\Psi_3(\ell)\rangle \right\}$.  To bound the gap of (\ref{othermathcalH}), we now consider the chain with a boundary qubit on each end, and again apply lemma \ref{ABC}.  We set $B = \textsf{I} \otimes (H \otimes \textsf{I}^{\otimes \bar{\bar{\ell}}} +\textsf{I}^{\lfloor \ell/{\bar{\ell}} \rfloor \bar{\ell}-1 } \otimes \bar{H} \otimes \textsf{I}^{\otimes \bar{\bar{\ell}}-1}+\textsf{I}^{\otimes \lfloor \ell/{\bar{\ell}} \rfloor \bar{\ell}} \otimes  H^R) \otimes \textsf{I}$ and $C = \textsf{I}^{\otimes \ell} \otimes \bar{H}^{0,1} + \bar{H}^{\ell,\ell+1} \otimes  \textsf{I}^{\otimes \ell}$ so that $A = B+C = {\mathcal H}$.  We know that $B$ is gapped and $\norm{C}_2 \le 3$.  Thus, we need to evaluate $C$ in the basis of ground states of $B$ and show the resulting matrix is gapped.  Lemma \ref{4states} shows that $B$ has $16$ zero-energy eigenstates.  A suitable basis is $
\left\{ |\Psi_0 (\ell)\rangle ,|\Psi_1(\ell)\rangle ,|\Psi_2(\ell)\rangle ,|\Psi_3(\ell)\rangle \right\}\otimes \left\{ \ket{0},\ket{1}\right\}$.  To determine the $16 \times 16$ matrix $P^BCP^B$, first focus on the $\bar{H}^{0,1}$ end of the chain.  At $\theta = 0$, the calculation reduces to computing $\textsf{I}\otimes \textsf{I} \otimes {H}^{B}$ in the basis 
\begin{align*}
\left\{ \ket{0},\ket{1}\right\}  \otimes  & \left\{  \frac{\ket{0}   \ket{0} +  \ket{1}   \ket{1}}{\sqrt{2}} , \frac{\ket{0} \ket{1}+\ket{1} \ket{0}}{\sqrt{2}},  \frac{\ket{0} \ket{1}-\ket{1} \ket{0}}{\sqrt{2}},  \frac{\ket{0}   \ket{0} -  \ket{1}   \ket{1}}{\sqrt{2}} \right\}  \\
& \hspace{3.25in} \otimes \left\{ \ket{0},\ket{1}\right\}.
\end{align*}
The resulting matrix is
\begin{align*}
\bar{h}_0\vert_{\theta = 0} = \tilde{h}_0 \vert_{\theta = 0} =   \left[\begin{array}{cc}1 & 0 \\ 0 & 1 \end{array} \right] \otimes \left[
\begin{array}{cccccccc}
 \frac{3}{4} & 0 & 0 & -\frac{1}{4} & 0 & -\frac{1}{4} & -\frac{1}{4} & 0 \\
 0 & \frac{3}{4} & -\frac{1}{4} & 0 & \frac{1}{4} & 0 & 0 & \frac{1}{4} \\
 0 & -\frac{1}{4} & \frac{3}{4} & 0 & \frac{1}{4} & 0 & 0 & \frac{1}{4} \\
 -\frac{1}{4} & 0 & 0 & \frac{3}{4} & 0 & -\frac{1}{4} & -\frac{1}{4} & 0 \\
 0 & \frac{1}{4} & \frac{1}{4} & 0 & \frac{3}{4} & 0 & 0 & -\frac{1}{4} \\
 -\frac{1}{4} & 0 & 0 & -\frac{1}{4} & 0 & \frac{3}{4} & -\frac{1}{4} & 0 \\
 -\frac{1}{4} & 0 & 0 & -\frac{1}{4} & 0 & -\frac{1}{4} & \frac{3}{4} & 0 \\
 0 & \frac{1}{4} & \frac{1}{4} & 0 & -\frac{1}{4} & 0 & 0 & \frac{3}{4} \\
 \end{array}
\right] .
\end{align*}
For general $\theta$, we define
\[
\tilde{h}_0 = \frac{\cos^2 \theta \left(\cos ^2\theta+\frac{\sin ^2\theta}{4}\right)^{\ell-1}}{(\cos ^2 \theta + (1/4) \sin ^2 \theta)^{\ell} - ((1/4) \sin ^2  \theta)^{\ell}}  \tilde{h}_0\vert_{\theta = 0}.
\]
This is analogous to (\ref{tildeh}), except that the coefficient in front is modified since  $\ell$ replaces $\bar{\ell}$ in $|\Psi_i (\ell)\rangle$ and the qubit at the end of the chain does not contribute any multiplicative factors.
Setting
\[
\frac{\Delta}{4} = 1- \sqrt{\frac{(\cos ^2 \theta + (1/4) \sin ^2 \theta)^{\ell} - ((1/4) \sin ^2  \theta)^{\ell}}{(\cos ^2 \theta + (1/4) \sin ^2 \theta)^{\ell} +3 ((1/4) \sin ^2  \theta)^{\ell}}}
\]
yields
\begin{align*}
\bar{h}_0 = &  \left(\left[\begin{array}{cc}1 & 0 \\ 0 & 1 \end{array} \right] \otimes \left[\begin{array}{cccc} 1-\frac{\Delta(\ell)}{4} & 0 & 0 & 0 \\ 0 & 1 & 0 & 0 \\ 0 & 0 & 1 & 0  \\ 0 & 0 & 0 &1 \end{array} \right] \otimes \left[\begin{array}{cc}1 & 0 \\ 0 & 1 \end{array} \right]\right) \tilde{h}_0 \\
& \hspace{0.75in} \left(\left[\begin{array}{cc}1 & 0 \\ 0 & 1 \end{array} \right] \otimes \left[\begin{array}{cccc} 1-\frac{\Delta(\ell)}{4} & 0 & 0 & 0 \\ 0 & 1 & 0 & 0 \\ 0 & 0 & 1 & 0  \\ 0 & 0 & 0 &1 \end{array} \right] \otimes \left[\begin{array}{cc}1 & 0 \\ 0 & 1 \end{array} \right] \right).
\end{align*}
A matrix $\bar{h}_{\ell+1}$ is obtained analogously for the other end of the chain:
\begin{align*}
\bar{h}_{\ell+1} & =   \left(\left[\begin{array}{cc}1 & 0 \\ 0 & 1 \end{array} \right] \otimes \left[\begin{array}{cccc} 1-\frac{\Delta(\ell)}{4} & 0 & 0 & 0 \\ 0 & 1 & 0 & 0 \\ 0 & 0 & 1 & 0  \\ 0 & 0 & 0 &1 \end{array} \right] \otimes \left[\begin{array}{cc}1 & 0 \\ 0 & 1 \end{array} \right]\right) \tilde{h}_{\ell+1} \\
& \hspace{0.5in} \left(\left[\begin{array}{cc}1 & 0 \\ 0 & 1 \end{array} \right] \otimes \left[\begin{array}{cccc} 1-\frac{\Delta(\ell)}{4} & 0 & 0 & 0 \\ 0 & 1 & 0 & 0 \\ 0 & 0 & 1 & 0  \\ 0 & 0 & 0 &1 \end{array} \right] \otimes \left[\begin{array}{cc}1 & 0 \\ 0 & 1 \end{array} \right] \right).
\end{align*}
with
\begin{align*}
\bar{h}_{\ell+1}\vert_{\theta = 0} = \tilde{h}_{\ell+1} \vert_{\theta = 0} =   \left[
\begin{array}{cccccccc}
 \frac{3}{4} & 0 & 0 & -\frac{1}{4} & 0 & 0 & -\frac{1}{2 \sqrt{2}} & 0 \\
 0 & \frac{1}{2} & 0 & 0 & -\frac{1}{2 \sqrt{2}} & 0 & 0 & \frac{1}{2 \sqrt{2}} \\
 0 & 0 & 1 & 0 & 0 & 0 & 0 & 0 \\
 -\frac{1}{4} & 0 & 0 & \frac{3}{4} & 0 & 0 & -\frac{1}{2 \sqrt{2}} & 0 \\
 0 & -\frac{1}{2 \sqrt{2}} & 0 & 0 & \frac{3}{4} & 0 & 0 & \frac{1}{4} \\
 0 & 0 & 0 & 0 & 0 & 1 & 0 & 0 \\
 -\frac{1}{2 \sqrt{2}} & 0 & 0 & -\frac{1}{2 \sqrt{2}} & 0 & 0 & \frac{1}{2} & 0 \\
 0 & \frac{1}{2 \sqrt{2}} & 0 & 0 & \frac{1}{4} & 0 & 0 & \frac{3}{4} \\
\end{array}
\right] \otimes  \left[\begin{array}{cc}1 & 0 \\ 0 & 1 \end{array} \right]
\end{align*}
and
\[
\tilde{h}_{\ell+1} = \frac{\cos^2 \theta \left(\cos ^2\theta+\frac{\sin ^2\theta}{4}\right)^{\ell-1}}{(\cos ^2 \theta + (1/4) \sin ^2 \theta)^{\ell} - ((1/4) \sin ^2  \theta)^{\ell}}  \tilde{h}_{\ell+1}\vert_{\theta = 0}.
\]
We find that $P^BCP^B = \bar{h}_0 + \bar{h}_{\ell+1}$ has a single zero-energy eigenstate, which is the zero-energy eigenstate of ${\mathcal H}$.  The gap of $P^BCP^B$, $c_1$, can be bounded using the eigenvalue stability inequality.  Direct diagonalization shows that the gap of $\tilde{h}_0 + \tilde{h}_{\ell+1}$ is 
\[
\frac{\cos^2 \theta \left(\cos ^2\theta+\frac{\sin ^2\theta}{4}\right)^{\ell-1}}{(\cos ^2 \theta + (1/4) \sin ^2 \theta)^{\ell} - ((1/4) \sin ^2 \theta)^{\ell}}
\] 
implying
\begin{align*}
c_1 \ge  \left(1- \frac{\Delta(\ell)}{2} \right)  \frac{\cos^2 \theta \left(\cos ^2\theta+\frac{\sin ^2\theta}{4}\right)^{\ell-1}}{(\cos ^2 \theta + (1/4) \sin ^2 \theta)^{\ell} - ((1/4) \sin ^2  \theta)^{\ell}}.
\end{align*}
Applying lemma \ref{ABC}, we find that ${\mathcal H}$ is gapped.
\end{proof}

\section{Swap chain Hamiltonian}
\label{sec1}

\subsection{Description of and motivation for the swap chain}
\label{4SS1} 

Quantum teleportation is a canonical method for transferring quantum information.  But a more direct approach is simply a swap gate between two qubits.  Similarly, the teleportation chain Hamiltonian has a streamlined cousin, the swap chain Hamiltonian, that we define in this section.

\begin{definition}. Let $\left\{\ket{0},\ket{1},\ket{\scriptscriptstyle IDLE}\right\}$ constitute an orthonormal basis of $\mathbb{C}^{3}$.
For $b \in \mathbb{F}_{2}$ define 
\[ 
| \psi_{b} \rangle = \cos \theta | 0 \rangle \otimes |b \rangle + \sin \theta | b \rangle \otimes |{\scriptscriptstyle IDLE} \rangle.
\]
Observe that 
\begin{align*}
& \left\{ \sin \theta | 0 \rangle \otimes | 0 \rangle - \cos \theta | 0  \rangle \otimes |{\scriptscriptstyle IDLE} \rangle,\sin \theta | 0 \rangle \otimes | 1 \rangle - \cos \theta | 1  \rangle \otimes |{\scriptscriptstyle IDLE} \rangle, \right. \\
&\hspace{2.9in} \left. | 1 \rangle \otimes |0 \rangle, | 1 \rangle \otimes |1 \rangle \right\}   \\
& \hspace{0.85in} \cup  \left\{ | \psi_{0} \rangle, | \psi_{1} \rangle, | {\scriptscriptstyle IDLE} \rangle \otimes |0 \rangle, | {\scriptscriptstyle IDLE} \rangle \otimes |1 \rangle, | {\scriptscriptstyle IDLE} \rangle \otimes |{\scriptscriptstyle IDLE}  \rangle \right\}
\end{align*}
is an orthornormal basis for $\mathbb{C}^{9}$, and let 
\begin{align}
\bar{H} = & ( \sin \theta | 0 \rangle \otimes | 0 \rangle - \cos \theta | 0  \rangle \otimes |{\scriptscriptstyle IDLE} \rangle ) ( \sin \theta \langle 0 | \otimes \langle 0 | - \cos \theta \langle 0  | \otimes \langle {\scriptscriptstyle IDLE} | ) \nonumber \\
 + &( \sin \theta | 0 \rangle \otimes | 1 \rangle - \cos \theta | 1  \rangle \otimes |{\scriptscriptstyle IDLE} \rangle ) ( \sin \theta \langle 0 | \otimes \langle 1 | - \cos \theta \langle 1  | \otimes \langle {\scriptscriptstyle IDLE} | ) \nonumber  \\
 +  & \ket{1} \bra{1} \otimes  \ket{0} \bra{0} +  \ket{1} \bra{1} \otimes  \ket{1} \bra{1}
\label{barHswap}
\end{align}
denote an orthogonal projection onto the first four vectors of the basis.   Then (\ref{H}) defines a Hamiltonian with open boundary conditions acting on $\lfloor \ell/{\bar{\ell}} \rfloor \bar{\ell}$  qutrits.  We attach a remnant and add a boundary qutrit to each end, and the swap chain Hamiltonian is then (\ref{mathcalH}) using $\bar{H}^{0,1} = \textsf{I} \otimes \ket{1}\bra{1} + \bar{H}$ and $\bar{H}^{\ell,\ell+1} = \ket{\scriptscriptstyle IDLE} \langle {\scriptscriptstyle IDLE} | \otimes \textsf{I} + \bar{H}$.  
\label{setup def}
\end{definition}
\begin{remark} To motivate the name swap chain Hamiltonian, consider a simplified Hamiltonian $\textsf{I} \otimes \ket{1}\bra{1} + \bar{H} + \ket{\scriptscriptstyle IDLE} \langle {\scriptscriptstyle IDLE} | \otimes \textsf{I} $ acting on 2 qutrits.   
\noindent By construction, the null space of $\bar{H}$ is $
\left\{ | \psi_{0} \rangle, | \psi_{1} \rangle, | {\scriptscriptstyle IDLE} \rangle   |0 \rangle, | {\scriptscriptstyle IDLE} \rangle   |1 \rangle, | {\scriptscriptstyle IDLE} \rangle   |1 \rangle \right\}$.  (Here and henceforth, we sometimes drop the tensor product $\otimes$ notation between states when no confusion will arise.)
Furthermore, since $\ket{\scriptscriptstyle IDLE} \langle {\scriptscriptstyle IDLE} | \otimes \textsf{I}$ raises the energy of the last three states
and $\langle \psi_{1} | \textsf{I} \otimes \ket{1}\bra{1} | \psi_{1} \rangle = \cos^{2} \theta$, we see that the unique ground state of the Hamiltonian is
$| \psi_{0} \rangle = \cos \theta | 0 \rangle | 0 \rangle + \sin \theta | 0 \rangle  | {\scriptscriptstyle IDLE} \rangle$.  Treat $| 0 \rangle$ and $| 1 \rangle$ as computational states and $\ket{\scriptscriptstyle IDLE}$ as a post-measurement state. Then, we can view $| \psi_{0} \rangle$ as being a coherent superposition of the state of two qubits before and after running through the circuit of Fig. \ref{swapcircuit} (in the case $| \mu \rangle = | 0 \rangle$).

\begin{figure}[H]
\begin{center}
\includegraphics[width=4.5in]{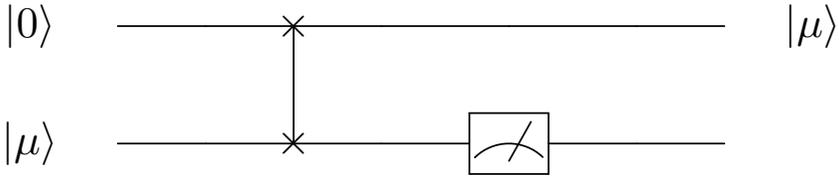}
\end{center}
\caption{Swap gate followed by measurement of lower qubit.}
\label{swapcircuit}
\end{figure}
We think of $\bar{H} $  as describing a single qubit going through one iteration of this swap circuit.  Increasing $\ell$ then amounts to running through more iterations of the circuit.  The term $\textsf{I} \otimes \ket{1}\bra{1}$ in $\bar{H}^{0,1}$ privileges starting in the computational state $| 0 \rangle$ as opposed to $| 1 \rangle$ (i.e., setting the initial $| \mu \rangle = | 0 \rangle$) while the term $\ket{\scriptscriptstyle IDLE} \langle {\scriptscriptstyle IDLE} | \otimes \textsf{I}$ in $\bar{H}^{\ell,\ell+1}$ penalizes ending in the non-computational state $\ket{\scriptscriptstyle IDLE}$.  Each $\bar{H}$ then enforces a version of the circuit with amplitude depending on $\theta$.  Altogether, we think of ${\mathcal H}$ as emulating the behavior of swapping a single qubit down a line.
\end{remark}

As in the case of the teleportation chain Hamiltonian, our application of the local gap method \cite{Knabe1988} or the Martingale method \cite{Nachtergaele1996} did not succeed in establishing gapped behavior over the entire range $0 \le \theta < \pi/2$.  We therefore try the  renormalization method.  To facilitate the application of the renormalization method, we define some matrices that will show up repeatedly during our analysis.  We also determine the spectra of these matrices.

\begin{definition}
Set $s = \sin \theta$ and $t = \tan \theta$ and for $m \geq 3$ define the $m \times m$ matrix
\[
V (m) = 
\left(
\begin{array}{cccccccc}
t^2         & -s  & -s^2 & -s^3 & \hdots & -s^{m-3} & -s^{m-2} & -ts^{m-2}  \\
-s       & 1    & 0       & 0       & \hdots & 0           & 0           & 0          \\
-s^2     & 0      & 1     & 0       & \hdots & 0           & 0           & 0          \\
-s^3     & 0      & 0       & 1     & \hdots & 0           & 0           & 0          \\
\vdots      & \vdots & \vdots  & \vdots  &        & \vdots      & \vdots      & \vdots     \\
-s^{m-3} & 0      & 0       & 0       & \hdots & 1         & 0           & 0          \\
-s^{m-2} & 0      & 0       & 0       & \hdots & 0           & 1         & 0          \\
-ts^{m-2}   & 0      & 0       & 0       & \hdots & 0           & 0           & 1          
\end{array}
\right).
\]
Let $W(m-1)$ be the matrix obtained from $V(m)$ by deleting the last row and column.
\end{definition}

\begin{proposition}
The spectra of $V(m)$ and $W(m-1)$, including multiplicities, are given by 
\[
 \left\{ 0 \right\} \cup \bigcup_{i=1}^{m-2} \left\{ 1 \right\} \cup \left\{ \sec^2 \theta \right\}
\]
and
\[
 \left\{ \frac{1 - \sqrt{1 - 4 \cos^{2} \theta \sin^{2(m-1)} \theta}}{2 \cos^{2} \theta} \right\} \cup \bigcup_{i=1}^{m-3} \left\{ 1 \right\} \cup \left\{ \frac{1 + \sqrt{1 - 4 \cos^{2} \theta \sin^{2(m-1)} \theta}}{2 \cos^{2} \theta} \right\}.
\]
In particular, these matrices have at most a single eigenvalue smaller than 1.
\label{special matrices}
\end{proposition}

\begin{proof}
We only prove the statement for $W(m-1)$; the proof for $V(m)$ is analogous. We first seek eigenvectors of the form $\left( \lambda, s, s^2,  \hdots, s^{m-3}, s^{m-2} \right)^{\dagger}$.  Considering any coordinate beyond the first, we see that the eigenvalue must be $1-\lambda$.  This constrains the first coordinate so that  
\[
\lambda(1-\lambda) = t^2 \lambda - \sum_{j=1}^{m-2} s^{2j} = t^2 \lambda - \left( \frac{s^{2} - s^{2(m-1)}}{1-s^2} \right).
\]
Solving this quadratic and calculating $1-\lambda$ leads to the first and last asserted eigenvalues.

In order to completely classify the spectrum of $W(m-1)$, we now provide $m-3$ linearly independent eigenvectors with eigenvalue 1.  We do this by simply assuring that things cancel out appropriately in the first coordinate.  More specifically, we choose vectors of the form 
\[
(0,s,-1,0,0,0,0,\hdots)^{\dagger}, (0,s^2,0,-1,0,0,0,\hdots)^{\dagger},  (0,s^3,0,0,-1,0,0,\hdots)^{\dagger}, \hdots
\]
The proof of the proposition is then complete after noting that 
\begin{align*}
1 + \sqrt{1 - 4 \cos^{2} \theta \sin^{2(m-1)} \theta} & \geq  1 + \sqrt{1 - 4 \cos^{2} \theta \sin^{2} \theta} \\
& = 1 + \left\vert \cos^2 \theta - \sin^2 \theta \right\vert \geq 2 \cos^2 \theta.
\end{align*}
\end{proof}

\subsection{Ground states of the swap chain}
\label{4SS2} 

Consider a segment of length ${\bar{\ell}} \geq 2$ but without the boundary projectors $ I \otimes \ket{1}\bra{1} $ and $ \ket{\scriptscriptstyle IDLE} \langle {\scriptscriptstyle IDLE} | \otimes I$.  We wish to completely describe its ground state space.  But before doing so, we define two special sequences of states and establish some useful technical properties that these states display.

\begin{definition}
We shall call $| \Psi_{0}(\bar{\ell};0) \rangle $ and $| \Psi_{1}(\bar{\ell};0) \rangle$ the history states of a length ${\bar{\ell}}$ segment.  Begin by setting $| \Psi_{b} (1;0) \rangle =  | b  \rangle$ for $b \in \mathbb{F}_{2}$.  Then for any ${\bar{\ell}} > 1$ we define  $| \Psi_{b} (\bar{\ell};0) \rangle$ recursively by 
\[
| \Psi_{b} (\bar{\ell}-1;0) \rangle = \alpha | 0 \rangle | \mu \rangle +  \beta | 1 \rangle | \nu \rangle \implies | \Psi_{b}({\bar{\ell}};0) \rangle = \alpha | \psi_{0} \rangle | \mu \rangle +  \beta | \psi_{1} \rangle | \nu \rangle,
\]
where $| \psi_{0} \rangle $ and $| \psi_{1} \rangle$ are the two-qutrit states described in Definition \ref{setup def}.  We also define projection operators onto these states by setting 
\[
\bar{P}_{b}(\bar{\ell};0) =  | \Psi_{b}(\bar{\ell};0) \rangle \langle \Psi_{b}(\bar{\ell};0) |, \hspace{.1in}\bar{P}(\bar{\ell};0) = \bar{P}_{0}(\bar{\ell};0)+ \bar{P}_{1}(\bar{\ell};0).
\]
\end{definition}

\begin{definition}
Define states $| \Psi_{b}(\bar{\ell};j) \rangle = | {\scriptscriptstyle IDLE} \rangle^{\otimes j} | \Psi_{b}({\bar{\ell}}-j,0) \rangle$  and projectors
\begin{align*}
&\bar{P}_{b}(\bar{\ell};1\cdots \bar{\ell}-1) = \sum_{j=1}^{{\bar{\ell}}-1} | \Psi_{b}(\bar{\ell};j) \rangle  \langle \Psi_{b}(\bar{\ell};j) |,  \hspace{.1in} \bar{P}(\bar{\ell};\bar{\ell})  =  \ket{\scriptscriptstyle IDLE}^{\otimes {\bar{\ell}}}  \bra{\scriptscriptstyle IDLE}^{\otimes {\bar{\ell}}},\\
&  \hspace{.1in} \bar{P}(\bar{\ell};1\cdots \bar{\ell}) = \bar{P}_{0}(\bar{\ell};1\cdots \bar{\ell}-1) + \bar{P}_{1}(\bar{\ell};1\cdots \bar{\ell}-1)+  \bar{P}(\bar{\ell}; \bar{\ell})
\end{align*}
for ${\bar{\ell}}>1$, $1 \leq j \leq \bar{\ell} -1$, and $b \in \mathbb{F}_{2}$.
\end{definition}

\begin{proposition}
Letting $| \psi \rangle = \cos \theta | 0 \rangle + \sin \theta \ket{\scriptscriptstyle IDLE}$, these history states enjoy the following properties:
\begin{enumerate}[(i)]
    \item When expressed in the basis $\left\{ | 0 \rangle, | 1 \rangle, \ket{\scriptscriptstyle IDLE} \right\}^{\otimes {\bar{\ell}}}$, no terms of $| \Psi_{0}(\bar{\ell};0) \rangle $ or $| \Psi_{1}(\bar{\ell};0) \rangle$ begin with $\ket{\scriptscriptstyle IDLE}$.
    \item   When expressed in the basis $\left\{ | 0 \rangle, | 1 \rangle, \ket{\scriptscriptstyle IDLE} \right\}^{\otimes {\bar{\ell}}}$, each term of $| \Psi_{1}(\bar{\ell};0) \rangle$ has a single $| 1 \rangle$ in it.
    \item $| \Psi_{0}(\bar{\ell};0) \rangle$ is separable with $| \Psi_{0}(\bar{\ell};0) \rangle = | 0 \rangle  | \psi \rangle^{\otimes {\bar{\ell}}-1}$.
    \item $| \Psi_{1}(\bar{\ell};0) \rangle =  \cos \theta | 0 \rangle  \sum_{j=0}^{\bar{\ell}-2} \sin^{j} \theta  | \psi \rangle^{\otimes \bar{\ell}-2-j} |1 \rangle \ket{\scriptscriptstyle IDLE}^{\otimes j}   + \sin^{\bar{\ell}-1} \theta | 1 \rangle \ket{\scriptscriptstyle IDLE}^{\otimes \bar{\ell}-1}$
     \item For any $1 \leq j \leq {\bar{\ell}}-1$ we have $\langle \Psi_{0}(\bar{\ell};0) | \left( | 0 \rangle \langle{\scriptscriptstyle IDLE} | \otimes \textsf{I}^{\otimes {{\bar{\ell}}-1}} \right)
    \ket{ \Psi_{0}(\bar{\ell};j)} = \cos \theta \sin^{j-1} \theta$.
    \item For any $1 \leq j \leq {\bar{\ell}}-1$ we have $\langle \Psi_{1}(\bar{\ell};0) | \left( | 0 \rangle \langle {\scriptscriptstyle IDLE} | \otimes \textsf{I}^{\otimes {{\bar{\ell}}-1}} \right)\ket{ \Psi_{1}(\bar{\ell};j)} = \cos \theta \sin^{j-1} \theta$.
    \item For $b \in \mathbb{F}_{2}$ we have $\langle \Psi_{b}(\bar{\ell};0) | \left( | 0 \rangle \langle {\scriptscriptstyle IDLE} | \otimes \textsf{I}^{\otimes {{\bar{\ell}}-1}} \right) \ket{\scriptscriptstyle IDLE}^{\otimes {\bar{\ell}}} = \sin^{{\bar{\ell}}-1} \theta \delta_{b,0}$.
\end{enumerate}
\label{properties of history states}
\end{proposition}

\begin{proof}
Properties (iii) and (iv) can be shown via simple inductive arguments and (i) and (ii) are immediate consequences.  The remaining three properties also follow from properties (iii) and (iv).   To show (vi), the only non-trivial one, we write
\begin{align*}
& \langle \Psi_{1}(\bar{\ell};0) | \left( | 0 \rangle \langle {\scriptscriptstyle IDLE} | \otimes \textsf{I}^{\otimes {{\bar{\ell}}-1}} \right)\ket{ \Psi_{1}(\bar{\ell};j)} = \\
& \hspace{0.25in} \left(  \cos \theta \langle 0 |  \sum_{j^\prime=0}^{\bar{\ell}-2} \sin^{j^\prime} \theta  \langle \psi |^{\otimes \bar{\ell}-2-j^\prime} \langle 1 | \bra{\scriptscriptstyle IDLE}^{\otimes j^\prime}  \right) \\ 
& \hspace{0.25in} \left( | 0 \rangle \langle {\scriptscriptstyle IDLE} | \otimes \textsf{I}^{\otimes {{\bar{\ell}}-1}} \right) \ket{{\scriptscriptstyle IDLE}}^{\otimes j}
 \left( \cos \theta | 0 \rangle  \sum_{j^{\prime\prime}=0}^{\bar{\ell}-j-2} \sin^{j^{\prime\prime}} \theta  | \psi \rangle^{\otimes \bar{\ell}-j-2-j^{\prime\prime}} |1 \rangle \ket{\scriptscriptstyle IDLE}^{\otimes j^{\prime\prime}}  \right.\\
&\hspace{2.5in}+ \sin^{\bar{\ell}-j-1} \theta | 1 \rangle \ket{\scriptscriptstyle IDLE}^{\otimes \bar{\ell}-j-1}\Bigg).
\end{align*}
The inner product between the sum over $j^\prime$ and the sum over $j^{\prime\prime}$ gets non-zero contributions when $j^\prime = j^{\prime\prime}$.   Similarly, the inner product between the sum over $j^\prime$ and the term $\ket{{\scriptscriptstyle IDLE}}^{\otimes j} \sin^{\bar{\ell}-j-1} \theta | 1 \rangle \ket{\scriptscriptstyle IDLE}^{\otimes \bar{\ell}-j-1}$ gets a non-zero contribution when $j^\prime = \bar{\ell}-j-1$.  We are left with
\begin{align*}
& \left(  \cos \theta   \sum_{j^{\prime\prime}=0}^{\bar{\ell}-j-2} \sin^{j^{\prime\prime}} \theta  \langle \psi |^{\otimes \bar{\ell}-2-j^{\prime\prime}} \right) \ket{{\scriptscriptstyle IDLE}}^{\otimes j-1}\cos \theta | 0 \rangle \sin^{j^{\prime\prime}} \theta  | \psi \rangle^{\otimes \bar{\ell}-j-2-j^{\prime\prime}} \\
& \hspace{1.0in} +  \cos \theta  \sin^{\bar{\ell}-j-1} \theta  \langle \psi |^{\otimes \bar{\ell}-2-(\bar{\ell}-j-1)} \ket{{\scriptscriptstyle IDLE}}^{\otimes j-1} \sin^{\bar{\ell}- j-1} \theta \\
& = \cos^3 \theta \sin^{j-1} \theta \sum_{j^{\prime\prime}=0}^{\bar{\ell}-j-2} \sin^{2j^{\prime\prime}} \theta + \cos \theta\sin^{j-1} \theta  \sin^{2(\bar{\ell}-j-1)} \theta \\
& = \cos^3 \theta  \sin^{j-1} \theta \frac{1 - \sin^{2(\bar{\ell}-j-1)} \theta}{1-\sin^2 \theta}+ \cos \theta\sin^{j-1} \theta  \sin^{2(\bar{\ell}-j-1)} \theta = \cos \theta  \sin^{j-1} \theta 
\end{align*}
as required.
\end{proof}

\begin{remark}
Initially, it may seem surprising that  $| \Psi_{0}(\bar{\ell};0) \rangle $ is separable whereas $| \Psi_{1}(\bar{\ell};0) \rangle $ is entangled.  However, there is an asymmetry between 0 and 1 resulting from the fact that our Hamiltonian in some sense prefers $| 0 \rangle$ as the state of the ``blank" qubits that we are swapping through.  Figure \ref{GroundStatesFig} may serve as a useful visualization for some of the notation and properties encountered thus far.
\end{remark}

\begin{figure}[H]
\begin{center}
\includegraphics[width=4.5in]{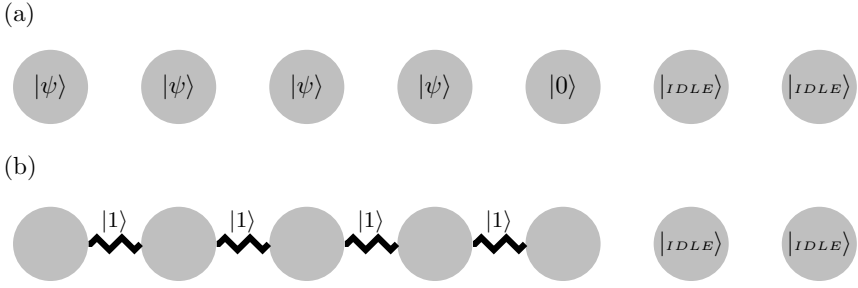}
\end{center}

\caption{When $\bar{\ell}=7$ and $j=2$ (a)  depicts $| \Psi_{0}(\bar{\ell};j) \rangle$  and (b) depicts $| \Psi_{1}(\bar{\ell};j) \rangle$.  Squiggly lines denote entanglement and the label $| 1 \rangle$ is meant to be an indication of property (ii).}
\label{GroundStatesFig}
\end{figure}

We now define a set that will turn out to form a basis for the ground state space of a length ${\bar{\ell}}$ segment without boundary projectors.  It is worth noting that this set differs substantially from the ground state space of the teleportation chain.  While the teleportation chain had a constant number of ground states $\bar{d} = 4$, the dimension of this set increases linearly as $\bar{d} = 2{\bar{\ell}}+1$.  We also set up a bit of additional notation that will be quite useful in our analysis moving forward.
\begin{definition}
For any ${\bar{\ell}} > 0$ we define a set $\Gamma_{{\bar{\ell}}}$ of cardinality $2{\bar{\ell}}+1$ by 
\[
\Gamma_{{\bar{\ell}}} = \bigcup_{j=0}^{{\bar{\ell}}-1} \left\{  | \Psi_{0}(\bar{\ell};j) \rangle \cup | \Psi_{1}(\bar{\ell};j) \rangle  \right\}  \cup \left\{ \ket{\scriptscriptstyle IDLE}^{\otimes {\bar{\ell}}} \right\}. 
\]
Observe that $\Gamma_{1} = \left\{ | 0 \rangle, | 1 \rangle, \ket{\scriptscriptstyle IDLE} \right\}$, i.e., our standard basis. We make note of the following partition of the identity operator acting on $\Gamma_{{\bar{\ell}}}$
\begin{align*}
\bar{P} & = \bar{P}(\bar{\ell};0) + \bar{P}(\bar{\ell};1\dots\bar{\ell}) \\
&= \bar{P}_{0}(\bar{\ell};0) + \bar{P}_{1}(\bar{\ell};0) + \bar{P}_{0}(\bar{\ell};1\dots\bar{\ell}-1) + \bar{P}_{1}(\bar{\ell};1\dots\bar{\ell}-1) + \bar{P}(\bar{\ell};\bar{\ell}).
\end{align*}
\end{definition}
We also define a set of transition operators that will appear in $\bar{h}$ and other Hamiltonian operators of interest.
\begin{definition}
\begin{align*}
& T_b(\bar{\ell};0,1\dots\bar{\ell}-1) = -\sum_{j=1}^{{\bar{\ell}}-1} \sin^{j} \theta \left( | \Psi_{b}(\bar{\ell};0) \rangle \langle \Psi_{b}(\bar{\ell};j)| + |\Psi_{b}(\bar{\ell};j) \rangle \langle \Psi_{b}(\bar{\ell};0) | \right), \\
& T(\bar{\ell};0,1\dots\bar{\ell}-1)= T_0(\bar{\ell};0,1\dots\bar{\ell}-1)  + T_1(\bar{\ell};0,1\dots\bar{\ell}-1) ,\\
& T_b(\bar{\ell};0,1\dots\bar{\ell}-2) = -\sum_{j=1}^{{\bar{\ell}}-2} \sin^{j} \theta \left( | \Psi_{b}(\bar{\ell};0) \rangle \langle \Psi_{b}(\bar{\ell};j)| + |\Psi_{b}(\bar{\ell};j) \rangle \langle \Psi_{b}(\bar{\ell};0) | \right), \\
& T(\bar{\ell};0,1\dots\bar{\ell}-2)= T_0(\bar{\ell};0,1\dots\bar{\ell}-2)  + T_1(\bar{\ell};0,1\dots\bar{\ell}-2) ,\\
& T_0(\bar{\ell};0,\bar{\ell}) = - \tan \theta \sin^{{\bar{\ell}}-1} \theta \left( | \Psi_{0}(\bar{\ell};0) \rangle \bra{\scriptscriptstyle IDLE}^{\otimes {\bar{\ell}}} + \ket{\scriptscriptstyle IDLE}^{\otimes {\bar{\ell}}} \langle \Psi_{0}(\bar{\ell};0) | \right),\\
& T_1(\bar{\ell};0,\bar{\ell})   =  - \cos \theta \sin^{{\bar{\ell}}} \theta \Big( \ket{0}\bra{1 } \otimes | \Psi_{1}(\bar{\ell};0) \rangle \bra{\scriptscriptstyle IDLE}^{\otimes {\bar{\ell}}} \\
& \hspace{2.0in} + \ket{1 } \bra{0} \otimes |{\scriptscriptstyle IDLE} \rangle^{\otimes {\bar{\ell}}} \langle \Psi_{1}(\bar{\ell};0) |  \Big),\\
& T_1(\bar{\ell};j;0,\bar{\ell})  =  - \cos \theta \sin^{{\bar{\ell}}} \theta \Big( |\Psi_{0}(\bar{\ell};j) \rangle \langle \Psi_{1}(\bar{\ell};j) | \otimes | \Psi_{1}(\bar{\ell};0) \rangle \bra{\scriptscriptstyle IDLE}^{\otimes {\bar{\ell}}} \\
& \hspace{2.0in} + |\Psi_{1}(\bar{\ell};j) \rangle \langle \Psi_{0}(\bar{\ell};j) | \otimes |{\scriptscriptstyle IDLE} \rangle^{\otimes {\bar{\ell}}} \langle \Psi_{1}(\bar{\ell};0) |  \Big).
\end{align*}

\end{definition}

Before proving that $\Gamma_{{\bar{\ell}}}$ is indeed a basis of ground states, we define an operator that describes $\bar{H} \otimes \textsf{I}^{\otimes{{\bar{\ell}}-2}}$ when one end is constrained to the space spanned by $\Gamma_{{\bar{\ell}}-1}$.  We then prove an important technical proposition that describes the structure of this operator in a fair bit of detail.

\begin{definition}
For any ${\bar{\ell}} \geq 2$ we define the operator $\bar{h}^{\Gamma_{{\bar{\ell}}-1}}$, which acts on $\Gamma_{1} \otimes \Gamma_{{\bar{\ell}}-1}$, by
\[
\bar{h}^{\Gamma_{{\bar{\ell}}-1}} = \sum_{| \mu \rangle \in \Gamma_{1} } \sum_{| \nu \rangle \in \Gamma_{{\bar{\ell}}-1} } \sum_{| \mu^\prime \rangle \in \Gamma_{1} } \sum_{| \nu^\prime \rangle \in \Gamma_{{\bar{\ell}}-1} } \left( \langle \mu |  \langle \nu | \bar{H} \otimes \textsf{I}^{\otimes{{\bar{\ell}}-2}} |  \mu^\prime \rangle  | \nu^\prime \rangle \right)  \left( | \mu \rangle  | \nu \rangle \langle  \mu^\prime | \langle \nu^\prime | \right).
\]
We will show that $\Gamma_{{\bar{\ell}}-1}$ is the ground state space of a segment of length ${\bar{\ell}}-1$; Fig. \ref{barhGamma} therefore provides an intuitive depiction of $\bar{h}^{\Gamma_{{\bar{\ell}}-1}}$.

\begin{figure}[H]
\begin{center}
\includegraphics[width=4.5in]{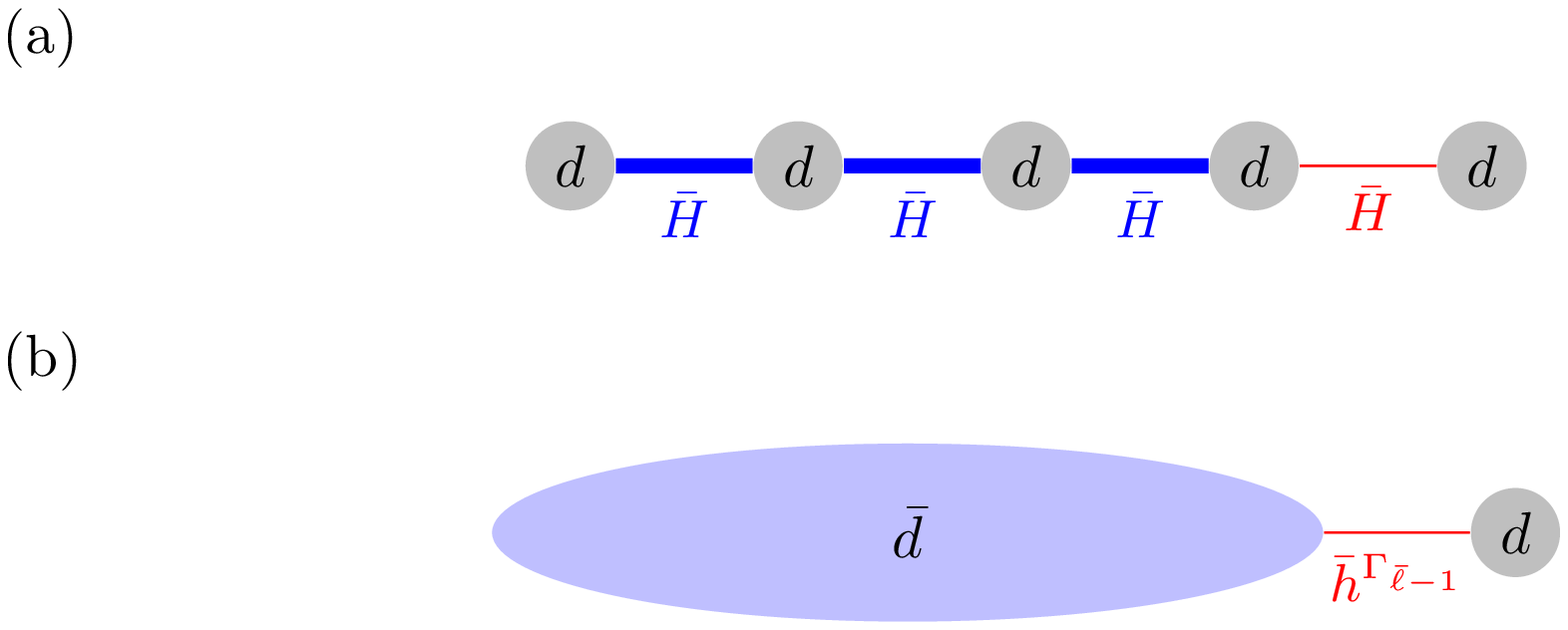}
\end{center}

\caption{$\bar{h}^{\Gamma_{{\bar{\ell}}-1}}$ in the case $\bar{\ell} = 5$.  We restrict our attention to the ground state space of operators represented by thick lines in (a) in order to arrive at the bond depicted in (b). }
\label{barhGamma}
\end{figure}

\end{definition}

\begin{proposition}
For ${\bar{\ell}} \geq 2$,
\begin{align*}
\bar{h}^{\Gamma_{{\bar{\ell}}-1}} & = | 0 \rangle \langle 0 | \otimes \cos^2 \theta \left( \tan^2 \theta \bar{P}(\bar{\ell}-1;0)  + \bar{P}(\bar{\ell}-1;1\dots\bar{\ell}-1)  \right. \\
& \left. + T(\bar{\ell}-1;0,1\dots\bar{\ell}-2) + T_0(\bar{\ell}-1;0,\bar{\ell}-1) \right)   \\
& + | 1 \rangle \langle 1 | \otimes\left( \bar{P}(\bar{\ell}-1;0)  + \cos^{2} \theta \bar{P}(\bar{\ell}-1;1\dots\bar{\ell}-1)  \right) +T_1(\bar{\ell}-1;0,\bar{\ell}-1).
\end{align*}
\label{matrix structure 1}
\end{proposition}

\begin{proof}
We begin by showing that there are no terms of the form $\ket{\scriptscriptstyle IDLE} \langle \mu^\prime | \otimes  | \nu \rangle \langle \nu^\prime |$ for any $| \mu^\prime \rangle \in \Gamma_{1}$ and $| \nu \rangle, | \nu^\prime \rangle \in \Gamma_{{\bar{\ell}}-1}$.  The matrix elements here look like $\bra{\scriptscriptstyle IDLE} \langle \nu | \bar{H} \otimes \textsf{I}^{\otimes{{\bar{\ell}}-2}} | \mu^\prime \rangle | \nu^\prime \rangle $, which must vanish since none of the rank one projectors defining 
\begin{align*}
\bar{H} = & ( \sin \theta | 0 \rangle \otimes | 0 \rangle - \cos \theta | 0  \rangle \otimes |{\scriptscriptstyle IDLE} \rangle ) ( \sin \theta \langle 0 | \otimes \langle 0 | - \cos \theta \langle 0  | \otimes \langle {\scriptscriptstyle IDLE} | ) \nonumber \\
 + &( \sin \theta | 0 \rangle \otimes | 1 \rangle - \cos \theta | 1  \rangle \otimes |{\scriptscriptstyle IDLE} \rangle ) ( \sin \theta \langle 0 | \otimes \langle 1 | - \cos \theta \langle 1  | \otimes \langle {\scriptscriptstyle IDLE} | ) \nonumber  \\
 +  & \ket{1} \bra{1} \otimes  \ket{0} \bra{0} +  \ket{1} \bra{1} \otimes  \ket{1} \bra{1}
\end{align*}
begins with $\ket{\scriptscriptstyle IDLE}$.  The same argument shows that there are no terms of the form $| \mu \rangle \bra{\scriptscriptstyle IDLE} \otimes  | \nu \rangle \langle \nu^\prime |$.

Next, we consider terms of the form $|0 \rangle \langle 1| \otimes | \nu \rangle \langle \nu^\prime |$ and note that the $|1 \rangle \langle 0|$ terms are handled by symmetry.  The matrix elements here look like $\langle 0 |  \langle \nu | \bar{H} \otimes \textsf{I}^{\otimes{{\bar{\ell}}-2}} | 1 \rangle | \nu^\prime \rangle $ for $| \nu \rangle, | \nu^\prime \rangle \in \Gamma_{{\bar{\ell}}-1}$.  Using the definition of $\bar{H}$ we rewrite these as
\[
- \left( \cos \theta \sin \theta \right) \langle \nu | \left( | 1 \rangle \bra{\scriptscriptstyle IDLE} \otimes \textsf{I}^{\otimes{{\bar{\ell}}-2}} \right) | \nu^\prime \rangle. 
\]
The only way $| \nu \rangle \in \Gamma_{{\bar{\ell}}-1}$ can start with a $|1 \rangle$ is if $| \nu \rangle = | \Psi_{1}(\bar{\ell}-1;0) \rangle$ and the only part of the state that accomplishes this looks like $\sin^{{\bar{\ell}}-2} \theta | 1 \rangle  \ket{\scriptscriptstyle IDLE}^{\otimes {\bar{\ell}}-2}$.  Thus, the matrix elements can be written as
\begin{align*}
- \left( \cos \theta \sin^{{\bar{\ell}}-1} \theta \right) \delta_{| \nu \rangle,  | \Psi_{1}(\bar{\ell}-1;0) \rangle }&  \bra{\scriptscriptstyle IDLE}^{\otimes {\bar{\ell}}-1} | \nu^\prime \rangle \\
& = - \left( \cos \theta \sin^{{\bar{\ell}}-1} \theta \right) \delta_{| \nu \rangle,  | \Psi_{1}(\bar{\ell}-1;0) \rangle }  \delta_{| \nu^\prime \rangle,  \ket{\scriptscriptstyle IDLE}^{\otimes {\bar{\ell}}-1 }}. 
\end{align*}
This produces $T_1(\bar{\ell}-1;0,\bar{\ell}-1)$ in $\bar{h}^{\Gamma_{{\bar{\ell}}-1}}$.

Now, we address terms of the form  $|1 \rangle \langle 1| \otimes | \nu \rangle \langle \nu^\prime |$.   The matrix elements here look like $\langle 1 |  \langle \nu | \bar{H} \otimes \textsf{I}^{\otimes{{\bar{\ell}}-2}} | 1 \rangle | \nu^\prime \rangle $.  Once again recalling the definition of $\bar{H}$, we express these as
\begin{align*}
 \langle \nu | \left( | 0 \rangle \langle0 | \otimes \textsf{I}^{\otimes{{\bar{\ell}}-2}} \right) | \nu^\prime \rangle & + \langle \nu | \left( | 1 \rangle \langle 1 | \otimes \textsf{I}^{\otimes{{\bar{\ell}}-2}} \right) | \nu^\prime \rangle \\
& + \cos^2 \theta   \langle \nu | \left( \ket{\scriptscriptstyle IDLE} \bra{\scriptscriptstyle IDLE} \otimes \textsf{I}^{\otimes{{\bar{\ell}}-2}} \right) | \nu^\prime \rangle.
\end{align*}
Using the properties in proposition \ref{properties of history states}, a simple but tedious calculation tells us that 
\begin{align*}
\langle \nu | \Big( | 0 \rangle \langle0 | & \otimes \textsf{I}^{\otimes{{\bar{\ell}}-2}} \Big) | \nu^\prime \rangle  + \langle \nu | \left( | 1 \rangle \langle 1 | \otimes \textsf{I}^{\otimes{{\bar{\ell}}-2}} \right) | \nu^\prime \rangle \\
&  = \delta_{| \nu \rangle,  | \Psi_{0}(\bar{\ell}-1;0) \rangle }  \delta_{| \nu^\prime \rangle,  | \Psi_{0}(\bar{\ell}-1;0) \rangle} + \delta_{| \nu \rangle,  | \Psi_{1}(\bar{\ell}-1;0) \rangle }  \delta_{| \nu^\prime \rangle,  | \Psi_{1}(\bar{\ell}-1;0) \rangle}
\end{align*}
Thus, $| 1 \rangle \langle 1 | \otimes\left( \bar{P}(\bar{\ell}-1;0)  + \cos^{2} \theta \bar{P}(\bar{\ell}-1;1\dots\bar{\ell}-1)  \right)$ appears in $\bar{h}^{\Gamma_{{\bar{\ell}}-1}}$.

Finally, consider the terms $|0 \rangle \langle 0| \otimes | \nu \rangle \langle \nu^\prime |$.   The matrix elements here take the form $\langle 0 |  \langle \nu | \bar{H} \otimes \textsf{I}^{\otimes{{\bar{\ell}}-2}} | 0 \rangle  | \nu^\prime \rangle $ for $| \nu \rangle, | \nu^\prime \rangle \in \Gamma_{{\bar{\ell}}-1}$.  These are a bit more complicated than in the previous cases.  Note that 
\begin{align*}
& \langle 0 |  \langle \nu | \bar{H} \otimes \textsf{I}^{\otimes{{\bar{\ell}}-2}} | 0 \rangle  | \nu^\prime \rangle  = \sin^{2} \theta \Big( \langle \nu | \left( | 0 \rangle \langle 0 | \otimes \textsf{I}^{\otimes{{\bar{\ell}}-2}} \right) | \nu^\prime \rangle   
\\
&\hspace{0.55in}  + \langle \nu | \left( | 1 \rangle \langle 1 | \otimes \textsf{I}^{\otimes{{\bar{\ell}}-2}} \right) | \nu^\prime \rangle \Big) + \cos^2 \theta   \langle \nu | \left( \ket{\scriptscriptstyle IDLE} \bra{\scriptscriptstyle IDLE} \otimes \textsf{I}^{\otimes{{\bar{\ell}}-2}} \right) | \nu^\prime \rangle 
\\
&\hspace{0.55in}  - \cos \theta \sin \theta \left(  \langle \nu | \left( | 0 \rangle \bra{\scriptscriptstyle IDLE} \otimes \textsf{I}^{\otimes{{\bar{\ell}}-2}} \right) | \nu^\prime \rangle +  \langle \nu | \left( \ket{\scriptscriptstyle IDLE} \langle 0 | \otimes \textsf{I}^{\otimes{{\bar{\ell}}-2}} \right) | \nu^\prime \rangle \right).
\end{align*}
Our analysis of the previous case implies that the $\sin^2 \theta$ term evaluates to
\begin{align*}
&\sin^2 \theta \left( \delta_{| \nu \rangle,  | \Psi_{0}(\bar{\ell}-1;0) \rangle }  \delta_{| \nu^\prime \rangle,  | \Psi_{0}(\bar{\ell}-1;0) \rangle} + \delta_{| \nu \rangle,  | \Psi_{1}(\bar{\ell}-1;0) \rangle }  \delta_{| \nu^\prime \rangle,  | \Psi_{1}(\bar{\ell}-1;0) \rangle} \right).
\end{align*}
The $\cos^2 \theta$ term is straightforward, so it only remains to determine the $\cos \theta \sin \theta$ term.  We will consider the matrix elements 
\[
- \cos \theta \sin \theta  \langle \nu | \left( | 0 \rangle \bra{\scriptscriptstyle IDLE} \otimes \textsf{I}^{\otimes{{\bar{\ell}}-2}} \right) | \nu^\prime \rangle ;
\]
the other term follows by Hermitian conjugation.  In order for such a matrix element to be non-zero it is necessary that $| \nu \rangle \in \left\{ | \Psi_{0}(\bar{\ell}-1;0) \rangle, | \Psi_{1}(\bar{\ell}-1;0) \rangle \right\}$.  We only consider the case where $| \nu \rangle = | \Psi_{0}(\bar{\ell}-1;0) \rangle$ as the other case is analogous. These terms vanish unless either $| \nu^{\prime} \rangle = | \Psi_{0}(\bar{\ell}-1;j) \rangle$ for some $1 \leq j \leq {\bar{\ell}}-2$ or $| \nu^{\prime} \rangle = \ket{\scriptscriptstyle IDLE}^{\otimes \bar{\ell}-1}$.  The values of the non-vanishing elements are given by properties (v), (vi), and (vii) of proposition \ref{properties of history states}.
\end{proof}

\begin{remark}
Since the notation utilized in proposition \ref{matrix structure 1} is a bit heavy, we provide an illustration in Fig. \ref{blocks1}.  This proposition amounts to establishing that $\bar{h}^{\Gamma_{{\bar{\ell}}-1}}$ has the block structure on the left.  In the basis adopted in (\ref{blockdiagonalized}) below, $\bar{h}^{\Gamma_{{\bar{\ell}}-1}}$ acquires the block structure on the right, where the two blocks in the upper left corner are both a constant times $V(\bar{\ell})$ from definition \ref{special matrices}. 
\end{remark}

\begin{figure}[H]
\begin{center}
\includegraphics[width=4.5in]{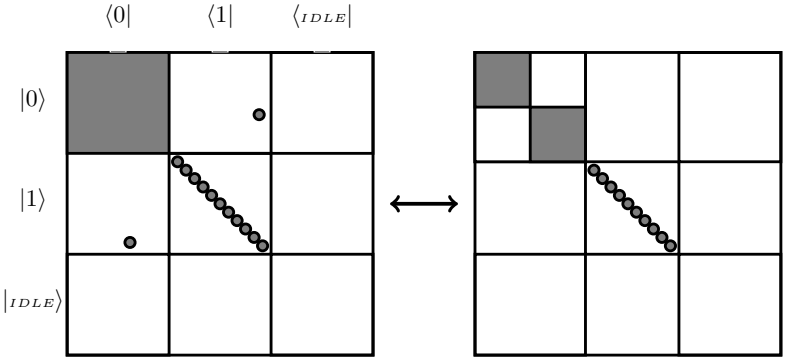}
\end{center}
\caption{The operator $\bar{h}^{\Gamma_{{\bar{\ell}}-1}}$  for two choices of basis.  The matrix on the right is obtained from that on the left by swapping two rows/columns in order to move the single non-zero element in the $|0 \rangle \langle 1|$ and  $|1 \rangle \langle 0|$ block near the bottom right corner of the $|0 \rangle \langle 0|$ block, and then possibly permuting basis vectors in the $|0 \rangle \langle 0|$ block, depending on their initial ordering.}
\label{blocks1}
\end{figure}
We are now ready to conclude this section with a lemma establishing the fact that $\Gamma_{{\bar{\ell}}}$ completely captures the ground state space of a length ${\bar{\ell}}$ segment without its boundary projectors.

\begin{lemma}
For any ${\bar{\ell}} \geq 2$, $\Gamma_{{\bar{\ell}}}$ is an orthonormal basis of ground states of $\bar{H}^S$. 
\label{the ground states}
\end{lemma}

\begin{proof}
We prove this using induction.  For ${\bar{\ell}}=2$,
\[
\Gamma_{2} =  \left\{ | \psi_{0} \rangle, | \psi_{1} \rangle, | {\scriptscriptstyle IDLE} \rangle \otimes |0 \rangle, | {\scriptscriptstyle IDLE} \rangle \otimes |1 \rangle, | {\scriptscriptstyle IDLE} \rangle \otimes |{\scriptscriptstyle IDLE}  \rangle \right\}
\]
is the set that is annihilated by (\ref{barHswap}), so suppose that our statement holds for ${\bar{\ell}}-1$.  The orthogonality of the states follows from properties (i), (ii), and (iii) of proposition \ref{properties of history states} so we begin by showing that $\Gamma_{{\bar{\ell}}}$ is comprised of zero eigenstates of $\bar{H}^S$.  First we consider states in $\left\{ \ket{\scriptscriptstyle IDLE} \right\} \otimes \Gamma_{{\bar{\ell}}-1}$.  These are annihilated by $\textsf{I}^{\otimes{{\bar{\ell}}-j-1}} \otimes \bar{H} \otimes \textsf{I}^{\otimes{j-1}}$ for $ 1\leq j \leq {\bar{\ell}}-2$ by the inductive hypothesis and by $\bar{H} \otimes \textsf{I}^{\otimes{{\bar{\ell}}-2}}$ since none of the projectors defining $\bar{H}$ begins with $\ket{\scriptscriptstyle IDLE}$.  Since
\[
\left\{ \ket{\scriptscriptstyle IDLE} \right\} \otimes \Gamma_{{\bar{\ell}}-1} =   \bigcup_{j=1}^{{\bar{\ell}}-1} \left\{  | \Psi_{0}(\bar{\ell};j) \rangle \cup | \Psi_{1}(\bar{\ell};j) \rangle  \right\}  \cup \left\{ \ket{\scriptscriptstyle IDLE}^{\otimes {\bar{\ell}}} \right\},
\]
it only remains to prove that the history states $| \Psi_{0}(\bar{\ell};0) \rangle $ and $| \Psi_{1}(\bar{\ell};0) \rangle$ are also ground states. We do this for $| \Psi_{0}(\bar{\ell};0) \rangle $ since the proof for $| \Psi_{1}(\bar{\ell};0) \rangle$ is identical.  $| \Psi_{0}(\bar{\ell};0) \rangle $ is annihilated by $\bar{H} \otimes \textsf{I}^{\otimes{{\bar{\ell}}-2}}$ by construction, since it begins in either $| \psi_{0} \rangle $ or $| \psi_{1} \rangle$, both of which are orthogonal to $\bar{H}$.  We now show that  $| \Psi_{0}(\bar{\ell};0) \rangle $ is annihilated by $\textsf{I}^{\otimes{{\bar{\ell}}-j-1}} \otimes \bar{H} \otimes \textsf{I}^{\otimes{j-1}}$ for $1 \le j \leq {\bar{\ell}}-2$.  Fix a generic $| \Psi_{0} (\bar{\ell}-1;0) \rangle = \alpha | 0 \rangle | \mu \rangle +  \beta | 1 \rangle | \nu \rangle$.  By the recursive definition of $| \Psi_{0} (\bar{\ell};0) \rangle$,
along with some straightforward manipulations, we can write
\begin{align*}
| \Psi_{0} (\bar{\ell};0) \rangle & = \alpha | \psi_{0} \rangle | \mu \rangle +  \beta |  \psi_{1} \rangle | \nu \rangle \\
& = 
\cos \theta | 0 \rangle | \Psi_{0} (\bar{\ell}-1;0) \rangle + \sin \theta \left( \alpha | 0 \rangle \ket{\scriptscriptstyle IDLE} | \mu \rangle  + \beta | 1 \rangle \ket{\scriptscriptstyle IDLE} | \nu \rangle \right). 
\end{align*}
The $\cos \theta$ term is annihilated by $\textsf{I}^{\otimes{{\bar{\ell}}-j-1}} \otimes \bar{H} \otimes \textsf{I}^{\otimes{j-1}}$ by the inductive hypothesis.  If $j < {\bar{\ell}}-2$ then the $\sin \theta$ term is also annihilated by the inductive hypothesis.  If $j = {\bar{\ell}}-2$ then it is instead annihilated by the fact that none of the rank one projectors defining $\bar{H}$ begins with $\ket{\scriptscriptstyle IDLE}$.

We must now show that these are the only zero eigenstates of $\bar{H}^S$. The inductive hypothesis tells us that $\Gamma_{{\bar{\ell}}-1}$ is a full orthonormal basis for the ground state space of a length ${\bar{\ell}}-1$ segment.  In light of this, and the fact that $\bar{H}^S$ is a sum of positive projection matrices, it suffices to show that the operator $\bar{h}^{\Gamma_{{\bar{\ell}}-1}}$ has $2{\bar{\ell}}+1$ zero eigenstates.  We know that the set $\left\{ \ket{\scriptscriptstyle IDLE} \right\} \otimes \Gamma_{{\bar{\ell}}-1}$ contributes $2{\bar{\ell}}-1$ of these so we must show that it has two more zero eigenstates when restricted to $\left\{| 0 \rangle, | 1 \rangle \right\} \otimes \Gamma_{{\bar{\ell}}-1}$.  For this, we use proposition \ref{matrix structure 1} in order to partition 
\begin{align}
& \bar{h}^{\Gamma_{{\bar{\ell}}-1}} =  \Big[ | 0 \rangle \langle 0 | \otimes \cos^2 \theta \left( \tan^2 \theta \bar{P}_{0}(\bar{\ell}-1;0)  + \bar{P}_{0}(\bar{\ell}-1;1\dots\bar{\ell}-2)  \right. \label{blockdiagonalized}\\
&\hspace{0.5in}  \left. + \bar{P}(\bar{\ell}-1;\bar{\ell}-1) 
+ T_0(\bar{\ell}-1;0,1\dots\bar{\ell}-2)  +T_0(\bar{\ell}-1;0,\bar{\ell}-1)  \right) \Big]  \nonumber\\
& +\Big[ | 0 \rangle \langle 0 | \otimes \cos^2 \theta \left( \tan^2 \theta \bar{P}_{1}(\bar{\ell}-1;0) +\bar{P}_{1}(\bar{\ell}-1;1\dots\bar{\ell}-2)   \right. \nonumber\\
&\hspace{0.25in}  \left. +  T_1(\bar{\ell}-1;0,1\dots\bar{\ell}-2)\right) + | 1 \rangle \langle 1 | \otimes \cos^2 \theta \bar{P}(\bar{\ell}-1;\bar{\ell}-1)  + T_1(\bar{\ell}-1;0,\bar{\ell}-1) \Big] \nonumber\\
& +  \Big[| 1 \rangle \langle 1 | \otimes\left( \bar{P}(\bar{\ell}-1;0)  + \cos^{2} \theta \bar{P}_{0}(\bar{\ell}-1;1\dots\bar{\ell}-2)  \right. \nonumber\\
& \left. \hspace{2.95in} + \cos^{2} \theta \bar{P}_{1}(\bar{\ell}-1;1\dots\bar{\ell}-2) \right)\Big].\nonumber
\end{align}
Each of the three bracketed expressions acts on a disjoint portion of the space and can be thought of as a distinct block.  The final bracketed expression clearly has no zero eigenstates over its part of the space.  Writing the first bracketed expression out in the basis
\[
\left\{ |0\rangle | \Psi_0(\bar{\ell}-1;0)\rangle \right\}\cup \bigcup_{j=1}^{{\bar{\ell}}-2} \left\{  | 0 \rangle | \Psi_0(\bar{\ell}-1;j) \rangle  \right\} \cup \left\{ | 0 \rangle \ket{\scriptscriptstyle IDLE}^{\otimes {\bar{\ell}}-1} \right\}.
\]
yields $\cos^2 \theta V({\bar{\ell}})$.  Writing the second bracketed expression out in the basis
\[
\left\{|0\rangle | \Psi_1(\bar{\ell}-1;0)\rangle \right\} \cup \bigcup_{j=1}^{{\bar{\ell}}-2} \left\{  | 0 \rangle | \Psi_1(\bar{\ell}-1;j) \rangle  \right\}   \cup \left\{ | 1 \rangle \ket{\scriptscriptstyle IDLE}^{\otimes {\bar{\ell}}-1} \right\}.
\]
yields $\cos^2 \theta V({\bar{\ell}})$ as well.  Proposition \ref{special matrices} shows that each block has a single zero eigenstate, completing the proof.
\end{proof}

\subsection{The renormalized swap chain and decaying correlations}
\label{4SS3}

In this section, we begin by defining $\bar{h}$ as in (\ref{barhfigure}), an operator that can be thought of as a two-sided version of  $\bar{h}^{\Gamma_{{\bar{\ell}}-1}}$.  We prove a technical proposition describing the structure of $\bar{h}$.  This will be analogous to our treatment of $\bar{h}^{\Gamma_{{\bar{\ell}}-1}}$.

\begin{proposition}
Letting $a_{j} = \cos^2 \theta$ for $0 \leq j \leq {\bar{\ell}}-2$, and $a_{{\bar{\ell}}-1}=1$, we have
\begin{align}  \label{barh}
\bar{h}  = \sum_{j=0}^{{\bar{\ell}}-1} & a_{j} \Big( | \Psi_{0}(\bar{\ell};j) \rangle   \langle \Psi_{0}(\bar{\ell};j) |\otimes \\
& \cos^2 \theta \left( \tan^2 \theta \bar{P}(\bar{\ell};0)+ \bar{P}(\bar{\ell};1\dots\bar{\ell}) + T(\bar{\ell};0,1\dots\bar{\ell}-1)  + T_0(\bar{\ell};0,\bar{\ell})  \right)  \nonumber \\ 
 +& \left. | \Psi_{1}(\bar{\ell};j) \rangle \langle \Psi_{1}(\bar{\ell};j) | \otimes\left( \bar{P}(\bar{\ell};0) + \cos^{2} \theta \bar{P}(\bar{\ell};1\dots\bar{\ell}) \right) + T_1(\bar{\ell};j;0,\bar{\ell})  \right). \nonumber
\end{align}
\label{matrix structure 2}
\end{proposition}

\begin{proof}
First, we note that there are no terms in $\bar{h}$ of the form $\ket{\scriptscriptstyle IDLE}^{\otimes {\bar{\ell}}} \langle \mu | \otimes | \nu \rangle \langle \nu^\prime|$ or $| \mu \rangle \bra{\scriptscriptstyle IDLE}^{\otimes {\bar{\ell}}} \otimes | \nu \rangle \langle \nu^\prime|$ for $| \mu \rangle, | \nu \rangle, | \nu^\prime \rangle \in \Gamma_{{\bar{\ell}}}$.   The matrix elements here vanish due to the fact that no terms in $\bar{H}$ begin with $\ket{\scriptscriptstyle IDLE}$.  Thus,  
\begin{align*}
\bar{h} = \sum_{\nu} \sum_{\nu^\prime}   \sum_{j=0}^{{\bar{\ell}}-1} \sum_{b=0}^1 \sum_{j^\prime=0}^{{\bar{\ell}}-1} \sum_{b'=0}^1   \Big( \langle \Psi_{b}(\bar{\ell};j) | \langle \nu |  \Big( \textsf{I}^{\otimes{{\bar{\ell}}-1}} \otimes \bar{H}& \otimes \textsf{I}^{\otimes{{\bar{\ell}}-1}} \Big)   |\Psi_{b^\prime}(\bar{\ell};j^\prime) \rangle | \nu^\prime \rangle\Big) \\
&  \left( | \Psi_{b}(\bar{\ell};j) \rangle | \nu \rangle \langle \Psi_{b^\prime}(\bar{\ell};j^\prime) | \langle \nu^\prime | \right).
\end{align*}
We claim that  $j \neq j^\prime$ implies $\langle \Psi_{b}(\bar{\ell};j) | \langle \nu | \left( \textsf{I}^{\otimes{{\bar{\ell}}-1}} \otimes \bar{H} \otimes \textsf{I}^{\otimes{{\bar{\ell}}-1}}  \right)  |\Psi_{b^\prime}(\bar{\ell};j^\prime) \rangle | \nu^\prime \rangle  = 0$.
WLOG assume that $j>j^\prime$.  We can cancel $j^\prime$ leading $\bra{\scriptscriptstyle IDLE}$'s on the left with $j^\prime$ leading $\ket{\scriptscriptstyle IDLE}$'s on the right to obtain
\[
\left( \bra{\scriptscriptstyle IDLE}^{\otimes j-j^\prime} \langle \Psi_{b}({\bar{\ell}}-j;0) | \langle \nu | \right) \left( \textsf{I}^{\otimes{{\bar{\ell}}-1-j^\prime}} \otimes \bar{H} \otimes \textsf{I}^{\otimes{{\bar{\ell}}-1}}  \right) \left( | \Psi_{b'}(\bar{\ell}-j^\prime;0) \rangle | \nu^\prime \rangle \right).
\]
If $j^\prime < {\bar{\ell}}-1$ then this vanishes due to property (i) in proposition \ref{properties of history states}.  If $j^\prime={\bar{\ell}}-1$ then $j=\ell$ and this instead vanishes since no terms in $\bar{H}$ begin with $\ket{\scriptscriptstyle IDLE}$.  This establishes an outer block structure of the form
\begin{align*}
\bar{h} =  \sum_{j=0}^{{\bar{\ell}}-1} &  \sum_{b=0}^1 \sum_{b'=0}^1 |\Psi_{b}(\bar{\ell};j) \rangle \langle \Psi_{b^\prime}(\bar{\ell};j) | \otimes  \\
&  \sum_{\nu} \sum_{\nu^\prime}  
\left( \langle \Psi_{b}(\bar{\ell};j) | \langle \nu | \left( \textsf{I}^{\otimes{{\bar{\ell}}-1}} \otimes \bar{H} \otimes \textsf{I}^{\otimes{{\bar{\ell}}-1}} \right)  |\Psi_{b^\prime}(\bar{\ell};j) \rangle | \nu^\prime \rangle\right) \left( | \nu \rangle \ \langle \nu^\prime | \right).
\end{align*}

The inner block structure (i.e., the fact that the non-zero blocks are as asserted) follows by the proof of proposition \ref{matrix structure 1}.  The extra factor of $\cos^{2} \theta$ in most of the blocks results from history states with ${\bar{\ell}}>1$ having amplitude $\sin \theta$ ending in $\ket{\scriptscriptstyle IDLE}$, which gets annihilated by the operator $\textsf{I}^{\otimes{{\bar{\ell}}-1}} \otimes \bar{H} \otimes \textsf{I}^{\otimes{{\bar{\ell}}-1}}$.
\end{proof}

\begin{remark}
Unpacking the notation a bit, proposition \ref{matrix structure 2} tells us that $\bar{h}$ is block diagonal with blocks of size $4{\bar{\ell}}+2$ and that each of these blocks looks like the non-trivial portion of the operator obtained in proposition \ref{matrix structure 1} up to a multiplicative positive constant.  In fact, all of the blocks are the same except the last one, which differs by a multiplicative constant $\tfrac{1}{\cos^2 \theta}$.  Finally, $\bar{h}$ is almost block diagonal with blocks of size $2{\bar{\ell}}+1$, seeing as the only impediment to this are the operators $T_1(\bar{\ell};j;0,\bar{\ell})$, which are exponentially small in ${\bar{\ell}}$.  These and a few other exponentially small elements will be discarded in order to form the $\tilde{h}$ that we desire in order to establish decaying correlations.  Fig. \ref{blocks2} may provide additional clarity.  
\end{remark}

\begin{figure}[H]
\begin{center}
\includegraphics[width=4.5in]{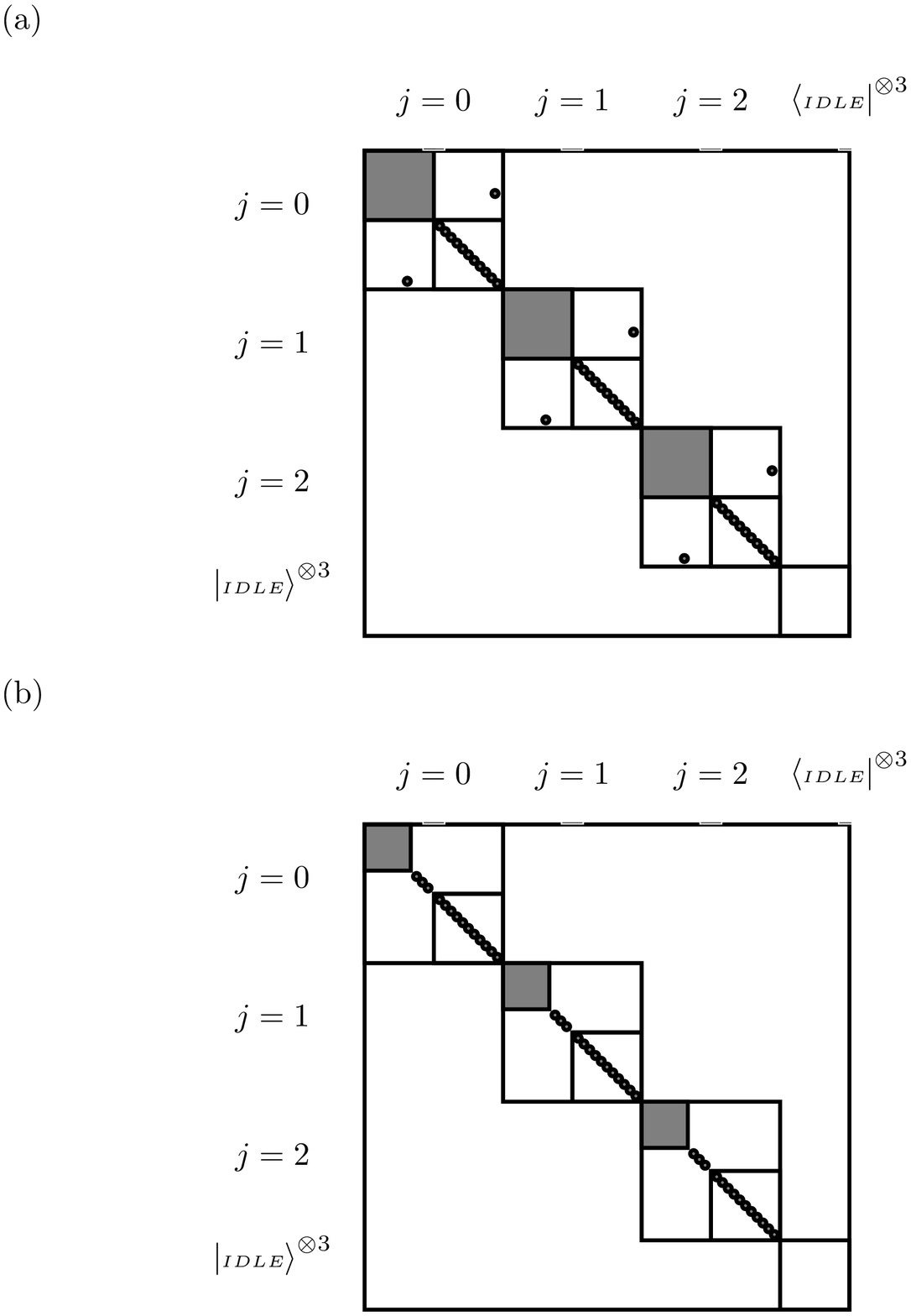}
\end{center}

\caption{(a) The operator $\bar{h}$ when $\bar{l}=3$.  The $j=2$ block differs from the $j=0$ and $j=1$ blocks by a multiplicative constant.  The blocks are like those displayed on the left of Fig. \ref{blocks1}.  \\ (b) The operator $\tilde{h}$ when $\bar{l}=3$.  It is obtained from $\bar{h}$ by dropping the solitary elements in the off-diagonal blocks as well as off-diagonal elements in three rows/columns of the solid blocks.}
\label{blocks2}
\end{figure}

\begin{proposition}
The smallest eigenvalue of $\bar{h}$ is 0 while the second smallest is $\cos^{4} \theta$.
\label{local gap}
\end{proposition}

\begin{proof}
Combine proposition \ref{matrix structure 2} with the last part of the proof of lemma \ref{the ground states} that further decomposes some of the blocks of $\bar{h}$ into sub-blocks.  We are left with either diagonal sub-blocks or sub-blocks that are constant multiples of $V({\bar{\ell}}+1)$.  The spectrum can then be determined with proposition \ref{special matrices}.  
\end{proof}

\begin{lemma}
Let $a_{j} = \cos^2 \theta$ for $0 \leq j \leq {\bar{\ell}}-2$, and $a_{{\bar{\ell}}-1}=1$, and define the operator
\begin{align*}
\tilde{h} = \sum_{j=0}^{{\bar{\ell}}-1} a_{j} \Big( | \Psi_{0}(\bar{\ell};j) \rangle &  \langle \Psi_{0}(\bar{\ell};j) | \otimes \\
& \cos^2 \theta \left( \tan^2 \theta \bar{P}(\bar{\ell};0) + \bar{P}(\bar{\ell};1\dots\bar{\ell}) + T(\bar{\ell};0,1\dots\bar{\ell}-2)  \right)  \\
+ | \Psi_{1}(\bar{\ell};j) \rangle & \langle \Psi_{1}(\bar{\ell};j) |  \otimes\left( \bar{P}(\bar{\ell};0) + \cos^{2} \theta \bar{P}(\bar{\ell};1\dots\bar{\ell}) \right) \Big).
\end{align*}
Then $\tilde{h} \geq 0$ and satisfies $\left\Vert \bar{h}- \tilde{h}\right\Vert_{2} \leq 2 \cos \theta \sin^{{\bar{\ell}}-1} \theta$.  Furthermore, $\tilde{h}$ commutes with its neighbors, i.e., $\left[\tilde{h} \otimes \bar{P},\bar{P} \otimes \tilde{h}\right] = 0$.  Finally, $\tilde{h} \otimes \bar{P} + \bar{P} \otimes \tilde{h} $ has at most $6{\bar{\ell}}+1$ eigenvalues below $\cos^4 \theta$ in the $(2 \bar{\ell}+1)^3$ space of ground states onto which $\bar{P}^{\otimes 3}$ projects.
\label{the big guns}
\end{lemma}

\begin{proof}
During the course of this proof we set $\tilde{h} = \sum_j \sum_b a_{j} | \Psi_{b}(\bar{\ell};j) \rangle \langle \Psi_{b}(\bar{\ell};j) | \otimes Q_{b}$ for $Q_{0} = \cos^2 \theta \left( \tan^2 \theta \bar{P}(\bar{\ell};0) + \bar{P}(\bar{\ell};1\dots\bar{\ell}) + T(\bar{\ell};0,1\dots\bar{\ell}-2)  \right)$ and $Q_{1}= \bar{P}(\bar{\ell};0) + \cos^{2} \theta \bar{P}(\bar{\ell};1\dots\bar{\ell}) $.  We also employ the decomposition $Q_0 = \cos^2 \theta \left( R_0 +R_1 + R_{{\scriptscriptstyle IDLE}} \right)$ where
\[
R_b = \tan^2 \theta \bar{P}_{b}(\bar{\ell};0)  + \bar{P}_{b}(\bar{\ell};1\dots\bar{\ell})  - | \Psi_{b}(\bar{\ell};\bar{\ell}-1)\rangle \langle \Psi_{b}(\bar{\ell};\bar{\ell}-1)| + T_b(\bar{\ell};0,1\dots\bar{\ell}-2)
\]
\[
R_{{\scriptscriptstyle IDLE}} = | \Psi_{0}(\bar{\ell};\bar{\ell}-1)\rangle \langle \Psi_{0}(\bar{\ell};\bar{\ell}-1)| + | \Psi_{1}(\bar{\ell};\bar{\ell}-1)\rangle \langle \Psi_{1}(\bar{\ell};\bar{\ell}-1)| + \bar{P}(\bar{\ell};\bar{\ell}).
\]
Note that this effectively partitions $Q_{0}$ into three distinct blocks.  $R_0$ and $R_1$ each act on a space of dimension ${\bar{\ell}}-1$ while the remaining block $R_{{\scriptscriptstyle IDLE}}$  acts on a space of dimension three.  See Fig. \ref{blocks3} for a simple depiction of this decomposition.

\begin{figure}[H]
\begin{center}
\includegraphics[width=4.5in]{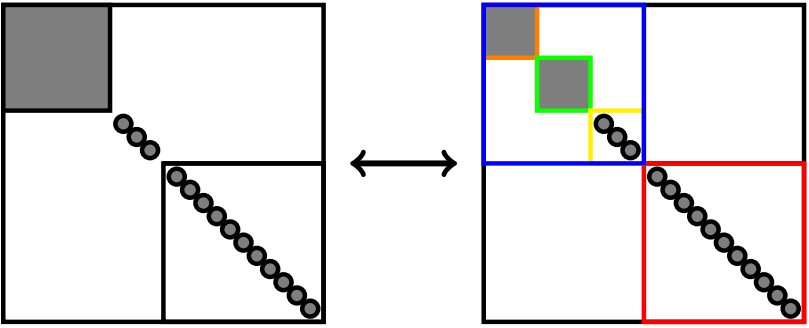}
\end{center}

\caption{Under an appropriate choice of basis, the nontrivial blocks of $\tilde{h}$ can be further decomposed as above.  The blue block represents $Q_{0}$ while the red block represents $Q_{1}$.   We observe that $Q_{0}$ can then be decomposed once again into $R_{0}$ (brown), $R_{1}$ (green), and $R_{{\scriptscriptstyle IDLE}}$ (yellow).}
\label{blocks3}
\end{figure}

We begin by showing that $\tilde{h} \geq 0$.  For later use, we will completely describe its spectrum.  By definition, it is block diagonal with $2{\bar{\ell}}$ non-zero blocks and a single zero block, contributing $2{\bar{\ell}}+1$ zero eigenvalues.  Of the $2\bar{\ell}$ non-zero blocks, ${\bar{\ell}}$ blocks have the form $a_{j} Q_1 $; they are all diagonal and have eigenvalues no smaller than $\cos^4 \theta$.  The remaining ${\bar{\ell}}$ blocks have the form $a_{j} Q_0 = a_j \cos^2 \theta \left( R_0 +R_1 + R_{{\scriptscriptstyle IDLE}} \right)$.  Each of these ${\bar{\ell}}$ blocks has two blocks of the form $a_{j} \cos^{2} \theta R_{b} = a_{j} \cos^{2} \theta W({\bar{\ell}}-1)$ and a diagonal block with entries $a_{j} \cos^{2} \theta$.  Thus, using proposition \ref{special matrices}, we know the entire spectrum.

Now set $\bar{k} = \bar{h} - \tilde{h}$.  This operator can be thought of as being block diagonal with blocks of size $4{\bar{\ell}}+2$.  This means that its operator norm is the max operator norm of its blocks.  In particular, 
\begin{align*}
 \left\Vert \bar{k} \right\Vert_{2} = \max_{j} |a_{j}|& \Big\Vert | \Psi_{0}(\bar{\ell};j) \rangle  \langle \Psi_{0}(\bar{\ell};j) |  \otimes \cos^2 \theta \Big(T_0(\bar{\ell};0,\bar{\ell})  \\
&  +  T(\bar{\ell};0,1\dots\bar{\ell}-1) -T(\bar{\ell};0,1\dots\bar{\ell}-2) \Big) + T_1(\bar{\ell};j;0,\bar{\ell}) \Big\Vert_{2} \\
& \leq 2 \cos \theta \sin^{{\bar{\ell}}-1} \theta,
\end{align*}
where we bounded the operator norm of this matrix by its Frobenius norm since it has only eight non-zero entries.

Next, we show commutativity.  A straightforward calculation allows us to write
\[
\left[\tilde{h} \otimes \bar{P},\bar{P} \otimes \tilde{h}\right] = \sum_{j=0}^{{\bar{\ell}}-1} \sum_{b=0}^{1} a_{j} | \Psi_{b}(\bar{\ell};j) \rangle \langle \Psi_{b}(\bar{\ell};j) | \otimes \left[ Q_{b} \otimes \bar{P}, \tilde{h} \right].
\]
In other words, the operators commute if and only if the corresponding blocks commute.  But
\[
\left[Q_{b'} \otimes \bar{P}, \tilde{h} \right] = \sum_{j=0}^{{\bar{\ell}}-1} \sum_{b=0}^{1} a_{j} \left[ Q_{b'}, | \Psi_{b}(\bar{\ell};j) \rangle \langle \Psi_{b}(\bar{\ell};j) | \right] \otimes Q_{b}
\]
since $\bar{P} Q_{b} = Q_{b}\bar{P}   = Q_{b}$.  When $b'=1$ this vanishes termwise since $Q_{1}$ is diagonal with respect to the basis $\Gamma_{\bar{\ell}}$.  When $b'=0$ we can calculate 
\begin{align*}
\left[Q_{0} \otimes \bar{P}, \tilde{h} \right] & = \cos^2 \theta \sum_{j=0}^{{\bar{\ell}}-1} \sum_{b=0}^{1} a_{j} \left[ T(\bar{\ell};0,1\dots\bar{\ell}-2), | \Psi_{b}(\bar{\ell};j) \rangle \langle \Psi_{b}(\bar{\ell};j) | \right] \otimes Q_{b} \\
& = \cos^4 \theta \sum_{b=0}^{1} \sum_{j=0}^{{\bar{\ell}}-2} \left[T_b(\bar{\ell};0,1\dots\bar{\ell}-2), | \Psi_{b}(\bar{\ell};j) \rangle \langle \Psi_{b}(\bar{\ell};j) | \right] \otimes Q_{b}, 
\end{align*}
where we only have to take the sum up to ${\bar{\ell}}-2$ from the definition of $T(\bar{\ell};0,1\dots\bar{\ell}-2)$, which causes $a_j = \cos^2 \theta$ for all the terms.  Pulling the sum over $j$ inside the commutator, we obtain the value 0, since the projector $\sum_{j=0}^{{\bar{\ell}}-2} | \Psi_{b}(\bar{\ell};j) \rangle \langle \Psi_{b}(\bar{\ell};j) |$ behaves just like the identity operator with respect to $T_b(\bar{\ell};0,1\dots\bar{\ell}-2)$.

Finally, we show that $\tilde{h} \otimes \bar{P} + \bar{P} \otimes \tilde{h}$ does not have too many small eigenvalues. We can write
\begin{align*}
\tilde{h} \otimes  \bar{P} & + \bar{P} \otimes \tilde{h} = \sum_{j=0}^{{\bar{\ell}}-1} \sum_{b=0}^1  a_{j} |\Psi_{b}(\bar{\ell};j) \rangle \langle \Psi_{b}(\bar{\ell};j) | \otimes Q_{b} \otimes \bar{P} \\
&  \hspace{0.5in} +  \sum_{j=0}^{{\bar{\ell}}-1} \sum_{b=0}^1 |\Psi_{b}(\bar{\ell};j) \rangle \langle \Psi_{b}(\bar{\ell};j) | \otimes \tilde{h} + |{\scriptscriptstyle IDLE} \rangle^{\otimes {\bar{\ell}}} \bra{\scriptscriptstyle IDLE}^{\otimes {\bar{\ell}}} \otimes \tilde{h} \\
                                                &  = \sum_{j=0}^{{\bar{\ell}}-1} \sum_{b=0}^1 |\Psi_{b}(\bar{\ell};j) \rangle   \langle \Psi_{b}(\bar{\ell};j) |  \otimes \left( a_{j} \left( Q_{b} \otimes \bar{P} \right) + \tilde{h} \right)  + |{\scriptscriptstyle IDLE} \rangle^{\otimes {\bar{\ell}}} \bra{\scriptscriptstyle IDLE}^{\otimes {\bar{\ell}}} \otimes \tilde{h}.
\end{align*}
This leaves us with $2{\bar{\ell}}+1$ blocks to analyze, which fall into the following three categories:
\begin{enumerate}
    \item ${\bar{\ell}}-1$ blocks of the form $\cos^2 \theta \left( Q_{0} \otimes \bar{P} \right) + \tilde{h}$ and one block of the form $ Q_{0} \otimes \bar{P}  + \tilde{h}$.
    \item ${\bar{\ell}}-1$ blocks of the form $\cos^2 \theta \left( Q_{1} \otimes \bar{P} \right) + \tilde{h}$ and one block of the form $ Q_{1} \otimes \bar{P}  + \tilde{h}$.
    \item One block of the form $\tilde{h}$.
\end{enumerate}

We begin with the block of type 3. We described the spectrum of $\tilde{h}$ while proving that $\tilde{h} \geq 0$ so this block contributes $2{\bar{\ell}}+1$ zero eigenvalues and at most $2{\bar{\ell}}$ more eigenvalues below $\cos^4 \theta$.  Now consider the blocks of type 2.  Since $\tilde{h} \geq 0$, the eigenvalues are all at least as large as those of $\cos^2 \theta \left( Q_{1} \otimes \bar{P} \right)$.  But $Q_{1}$ is diagonal with eigenvalues no smaller than $\cos^2 \theta$.  This means that, aside from blocks of type 1, we have at most $4{\bar{\ell}}+1$ eigenvalues below $\cos^4 \theta$.  Thus, it is sufficient to show that each block of type 1 contributes at most two eigenvalues below $\cos^4 \theta$.  We do this for a block of the form $\cos^2 \theta \left( Q_{0} \otimes \bar{P} \right) + \tilde{h}$ since $ Q_{0} \otimes \bar{P}  + \tilde{h}$ can be treated analogously.

We use our decomposition for $Q_0$ and a slightly modified one for $\tilde{h}$ in order to write
\begin{align*}
\cos^2 \theta \left( Q_{0} \otimes \bar{P} \right)+\tilde{h} = &\hspace{0.15in} \left( \cos^4 \theta R_0 \otimes \bar{P}+ \sum_{j=0}^{{\bar{\ell}}-2} \cos^2 \theta | \Psi_{0}(\bar{\ell};j) \rangle \langle \Psi_{0}(\bar{\ell};j) | \otimes Q_{0} \right) \\
&+  \left(\cos^4 \theta R_1  \otimes \bar{P} + \sum_{j=0}^{{\bar{\ell}}-2} \cos^2 \theta | \Psi_{1}(\bar{\ell};j) \rangle \langle \Psi_{1}(\bar{\ell};j) | \otimes Q_{1} \right) \\
&+  \left( \cos^4 \theta  R_{{\scriptscriptstyle IDLE}}  \otimes \bar{P}+ \sum_{b=0}^1  | \Psi_{b}(\bar{\ell};\bar{\ell}-1)\rangle \langle \Psi_{b}(\bar{\ell};\bar{\ell}-1)| \otimes Q_{b} \right). 
\end{align*}
Each of the three terms in parentheses forms a distinct block.  The third block has eigenvalues at least as large as those of $\cos^4 \theta R_{{\scriptscriptstyle IDLE}}$, since the other summand is positive semidefinite.  Similarly, the second block has eigenvalues at least as large as those of $\cos^2 \theta Q_1$.  These matrices are both diagonal and have a minimum eigenvalue of $\cos^4 \theta$.  Thus, it only remains to show that the first block has at most two eigenvalues that are smaller than $\cos^4 \theta$ and we will have concluded the proof.

The key point here is that both $\bar{P}$ and $\sum_{j=0}^{{\bar{\ell}}-2} | \Psi_{b}(\bar{\ell};j) \rangle \langle \Psi_{b}(\bar{\ell};j) |$ essentially act as the identity over their portion of the space, leading to a tensor decomposition of the eigenvectors of the first block.  Suppose that $\lambda_{1} \leq \lambda_{2} \leq ... \leq \lambda_{{\bar{\ell}}-1}$ are the eigenvalues of $\cos^4 \theta R_0$ with corresponding eigenvectors $| \lambda_{j} \rangle$.  Similarly, let $\mu_{1} \leq \mu_{2} \leq ... \leq \mu_{2{\bar{\ell}}+1}$ be the eigenvalues of $ \cos^2 \theta Q_{0} $ with corresponding eigenvectors $| \mu_{j} \rangle$.   Then the eigenvectors of the first block  are of the form  $| \lambda_{i} \rangle \otimes | \mu_{j} \rangle$, each with eigenvalue $ \lambda_{i} + \mu_{j}$.  Decomposing $\cos^2 \theta Q_0$ and repeatedly applying proposition \ref{special matrices},  we find that $\lambda_{2} \geq \cos^4 \theta$ and $\mu_{3}  \geq \cos^4 \theta$.  This means that the only eigenvalues that can be smaller than $\cos^4 \theta$ are $\lambda_{1} + \mu_{1}$ and $\lambda_{1} + \mu_{2}$.
\end{proof}

\subsection{Concluding that the swap chain is gapped}
\label{4SS4}

\begin{theorem}
For $0 \le \theta < \pi/2$, the swap chain Hamiltonian $\mathcal H$ is gapped with a unique ground state.
\end{theorem}
\begin{proof}
We begin by showing that an open chain of length $\ell$, which we shall denote by $\mathcal{H}_{0}^{\ell}$, is gapped.  Recalling the definition (\ref{mathcalH}) of $\mathcal H$, we observe that $\mathcal{H}_{0}^{\ell}$ satisfies
\[
{\mathcal H} = \bar{H}^{\ell,\ell+1} \otimes \textsf{I}^{\otimes \ell}+\textsf{I} \otimes  \mathcal{H}_{0}^{\ell} \otimes \textsf{I} +  \textsf{I}^{\otimes \ell} \otimes \bar{H}^{0,1}
\]
As described in remark \ref{DecompRemark}, we can further decompose $\mathcal{H}_{0}^{\ell}$ as
\[
\mathcal{H}_{0}^{\ell} = H \otimes  \textsf{I}^{\otimes \bar{\bar{\ell}}}+  \textsf{I}^{\otimes (\lfloor \ell/{\bar{\ell}} \rfloor \bar{\ell}-1)} \otimes \bar{H} \otimes \textsf{I}^{\otimes \bar{\bar{\ell}}-1} +\textsf{I}^{\otimes \lfloor \ell/{\bar{\ell}} \rfloor \bar{\ell}} \otimes  H^R.
\]
Refer to Fig. \ref{FancyHfigure} for an illustration.  

The first order of business is to invoke theorem \ref{Hgapped} to show that $H$ is gapped.  Lemma \ref{the big guns} shows that the chain exhibits decaying ground state correlations.  Using lemma \ref{the ground states} on a segment of length $3\bar{\ell}$, we conclude that the dimension of the kernel of $\bar{h} \otimes \bar{P} + \bar{P} \otimes \bar{h}$ is $z = 6 \bar{\ell}+1$.  Using, lemma \ref{the big guns} we can take $\tilde{g} = \cos^4 \theta$.  Moreover, the gap $\bar{g}$ is also $\cos ^4 \theta$ according to proposition \ref{local gap}.  It follows from theorem \ref{Hgapped} that $H$ of length $\lfloor \ell/\bar{\ell} \rfloor \bar{\ell}$ is gapped.

Next, we argue along the lines of remark \ref{mathcalHremark} to incorporate the remnant part of the chain and show that $\mathcal{H}_{0}^{\ell}$ is gapped.  Set $B = H \otimes \textsf{I}^{\otimes \bar{\bar{\ell}}} +\textsf{I}^{\otimes \lfloor \ell/{\bar{\ell}} \rfloor \bar{\ell}} \otimes  H^R$ and $C = \textsf{I}^{\lfloor \ell/{\bar{\ell}} \rfloor \bar{\ell}-1 } \otimes \bar{H} \otimes \textsf{I}^{\otimes \bar{\bar{\ell}}-1}$.  We have already shown that $H$ is gapped.  Observe that the gap of $H^{R}$ cannot shrink to zero with $\ell$ since $\bar{\bar{\ell}} < \bar{\ell}$, where $\bar{\ell}$ is some fixed value chosen independently of $\ell$.  Hence, $B$ is gapped.  Furthermore, it is clear that $\left\Vert C \right\Vert_{2} = 1$.  Thus, it suffices to show that $c_{1} = \lambda_{1}( P^{B} C P^{B} )$ is lower bounded by a positive constant, in order to employ lemma \ref{ABC} to conclude that $\mathcal{H}_{0}^{\ell} = B+C$ is gapped.  An analogous argument to the one used to prove proposition \ref{matrix structure 2} tells us that the matrix  $P^{B} C P^{B}$ looks quite similar to $\bar{h}$, except that now it has $\ell$ blocks of size $4\bar{\bar{\ell}}+2$ and a single zero block of size $2\bar{\bar{\ell}}+1$.  Thus, the same reasoning used to prove proposition \ref{local gap} tells us that $c_{1} \geq \cos^{4} \theta$.

Now, in order to show that $\mathcal H$ is gapped, we will use the strategy outlined in remark \ref{mathcalHremark} once again to account for the boundary terms. We begin by writing
\begin{align*}
\mathcal{H} & = \bar{H}^{\ell,\ell+1} \otimes \textsf{I}^{\otimes \ell}+\textsf{I} \otimes  \mathcal{H}_{0}^{\ell} \otimes \textsf{I} +  \textsf{I}^{\otimes \ell} \otimes \bar{H}^{0,1} \\
& = \ket{\scriptscriptstyle IDLE} \langle {\scriptscriptstyle IDLE} |  \otimes \textsf{I}^{\otimes \ell +1} + \bar{H}  \otimes \textsf{I}^{\otimes \ell} +\textsf{I} \otimes  \mathcal{H}_{0}^{\ell} \otimes \textsf{I} +  \textsf{I}^{\otimes \ell} \otimes \bar{H} +  \textsf{I}^{\otimes \ell + 1} \otimes \ket{1}\bra{1} \\
& = \ket{\scriptscriptstyle IDLE} \langle {\scriptscriptstyle IDLE} |  \otimes \textsf{I}^{\otimes \ell +1} + \mathcal{H}_{0}^{\ell+ 2}  +  \textsf{I}^{\otimes \ell + 1} \otimes \ket{1}\bra{1}.
\end{align*}
This time we set $B=\mathcal{H}_{0}^{\ell+ 2}$ and $C=\ket{\scriptscriptstyle IDLE} \langle {\scriptscriptstyle IDLE} |  \otimes \textsf{I}^{\otimes \ell +1}  +  \textsf{I}^{\otimes \ell + 1} \otimes \ket{1}\bra{1}$.  We have just established that $B$ is gapped and it is not hard to see that $\left\Vert C \right\Vert_{2} = 2$.  Thus, once again, we need only show that $c_{1} = \lambda_{1}( P^{B} C P^{B} )$ is lower bounded by a positive constant and then lemma \ref{ABC} will imply that $\mathcal{H}$ is gapped.  We know from lemma \ref{the ground states} that $\Gamma_{\ell+2}$ is the ground state space of $B$.  Recalling the definitions of the states that make up $\Gamma_{\ell+2}$, a straightforward calculation then allows us to write
\[
P^{B} \ket{\scriptscriptstyle IDLE} \langle {\scriptscriptstyle IDLE} | \otimes \textsf{I}^{\otimes \ell +1}  P^{B} =  \sum_{b=0}^{1} \sum_{j=1}^{\ell + 1} | \Psi_{b}(\ell+2;j) \rangle \langle \Psi_{b}(\ell+2;j) | + \left( \ket{\scriptscriptstyle IDLE} \langle {\scriptscriptstyle IDLE} |\right)^{\otimes \ell +2}
\]
as well as
\[
P^{B}  \textsf{I}^{\otimes \ell + 1} \otimes \ket{1}\bra{1}  P^{B} = \cos^{2} \theta \sum_{j=0}^{\ell+1} | \Psi_{1}(\ell+2;j) \rangle \langle \Psi_{1}(\ell+2;j) |.
\]
From this it is a simple task to write down the spectrum of $P^{B} C P^{B}$ when constrained to the space spanned by $\Gamma_{\ell+2}$:
\[
\begin{cases}
\lambda = 0 \text{ when } \ket{\lambda} =   | \Psi_{0}(\ell+2;0) \rangle\\
\lambda = \cos^2 \theta \text{ when } \ket{\lambda} =   | \Psi_{1} (\ell+2;0) \rangle \\
\lambda = 1 \text{ when } \ket{\lambda} \in \bigcup_{j=1}^{\ell+1} \left\{  | \Psi_{0}(\ell+2;j) \rangle \right\} \cup \ket{\scriptscriptstyle IDLE}^{\otimes \ell + 2}  \\
\lambda = \cos^2 \theta + 1  \text{ when } \ket{\lambda} \in \bigcup_{j=1}^{\ell+1} \left\{  | \Psi_{1}(\ell+2;j) \rangle \right\}.
\end{cases}
\]
Therefore, $c_{1}= \cos^{2} \theta$ and $\mathcal H$ is gapped.  Furthermore, since $\mathcal{H} = B + C$, such a spectral decomposition immediately implies that the ground state of $\mathcal H$ is unique.
\end{proof}

\section*{Acknowledgements}

We are grateful for helpful suggestions from Hosho Katsura.

\bibliography{Gapped} 


\end{document}